\documentclass[sigconf, nonacm]{acmart} 
\usepackage[labelfont=bf,textfont=md]{caption}
\usepackage{graphicx}
\usepackage{xspace}
\usepackage{comment}
\usepackage{epigraph} 
\usepackage{amsfonts}
\usepackage{soul} 
\usepackage{mathtools, nccmath}
\usepackage{mathrsfs}
\usepackage{algorithm} 
\usepackage{algpseudocode} 
\usepackage[utf8]{inputenc} 
\usepackage[T1]{fontenc}
\usepackage{siunitx} 
\usepackage{multirow}

\usepackage{comment}
\algnewcommand\algorithmicswitch{\textbf{switch}}
\algnewcommand\algorithmiccase{\textbf{case}}
\algnewcommand\algorithmicassert{\texttt{assert}}
\algnewcommand\Assert[1]{\State \algorithmicassert(#1)}%
\algdef{SE}[SWITCH]{Switch}{EndSwitch}[1]{\algorithmicswitch\ #1\ \algorithmicdo}{\algorithmicend\ \algorithmicswitch}%
\algdef{SE}[CASE]{Case}{EndCase}[1]{\algorithmiccase\ #1}{\algorithmicend\ \algorithmiccase}%
\algtext*{EndSwitch}%
\algtext*{EndCase}%
\usepackage{subcaption}
\usepackage{hyperref} 
\usepackage{amsthm}
\newtheorem{definition}{Definition}
\newtheorem{theorem}{Theorem}
\newtheorem{property}{Property}
\newtheorem{corollary}{Corollary}
\newtheorem*{problem}{Problem\textrm{ }Statement}
\usepackage{ragged2e}
\usepackage{wrapfig}

\usepackage[most]{tcolorbox}
\newtcolorbox{myboxi}[1][]{
  breakable,
  title=#1,
  colback=gray!10!white,
  colbacktitle=white,
  coltitle=black,
  fonttitle=\bfseries,
  bottomrule=0.5pt,
  toprule=0.5pt,
  leftrule=0.5pt,
  rightrule=0.5pt,
  titlerule=0pt,
  colframe=black,
  boxsep=1pt,left=2pt,right=2pt,top=2pt,bottom=2pt
}

\usepackage{listings}
\usepackage[title]{appendix}

\definecolor{codegreen}{rgb}{0,0.6,0}
\definecolor{codegray}{rgb}{0.5,0.5,0.5}
\definecolor{codepurple}{rgb}{0.58,0,0.82}
\definecolor{backcolour}{rgb}{0.95,0.95,0.92}

\lstdefinestyle{mystyle}{
  backgroundcolor=\color{backcolour},   
  commentstyle=\color{codegreen},
  keywordstyle=\color{magenta},
  numberstyle=\tiny\color{codegray},
  stringstyle=\color{codepurple},
  basicstyle=\ttfamily\footnotesize,
  breakatwhitespace=false,         
  breaklines=true,                 
  captionpos=b,                    
  keepspaces=true,                 
  numbersep=5pt,                  
  showspaces=false,                
  showstringspaces=false,
  showtabs=false,                  
  tabsize=2
}

\newcounter{enum}
\newenvironment{packed_enum}{
\begin{list}{\textbf{(\arabic{enum})}}{
  \setlength{\itemsep}{0pt}
  \setlength{\parskip}{0pt}
  \setlength{\labelwidth}{-5 pt}
  \setlength{\leftmargin}{0 pt}
  \setlength{\itemindent}{0pt}
  \setlength{\topsep}{0pt}
  \usecounter{enum}}
}{\end{list}}

\newcommand\vldbdoi{XX.XX/XXX.XX}

\newcommand\vldbavailabilityurl{URL_TO_YOUR_ARTIFACTS}
\newcommand\vldbpagestyle{plain} 

\newcommand{\header}[1]{\vspace{1mm}\noindent{\bfseries{#1.}}}

\newcommand{\wharf}[0]{\textsc{Wharf}\xspace}

\newcounter{MakisNOC}
\stepcounter{MakisNOC}

\newcounter{KaustubhNOC}
\stepcounter{KaustubhNOC}

\newcounter{ZoiNOC}
\stepcounter{ZoiNOC}

\newcounter{JorgeNOC}
\stepcounter{JorgeNOC}

\newcounter{EleniNOC}
\stepcounter{EleniNOC}

\newcounter{VolkerNOC}
\stepcounter{VolkerNOC}

\begin{document}
\title{Space-Efficient Random Walks on Streaming Graphs}






\author{Serafeim Papadias \ \ Zoi Kaoudi \ \ Jorge-Arnulfo Quian\'e-Ruiz \ \ Volker Markl}
\affiliation{\{s.papadias, zoi.kaoudi, jorge.quiane, volker.markl\}@tu-berlin.de}
\affiliation{
    \institution{Technische Universit\"at Berlin}
}

\begin{abstract}
Graphs in many applications, such as social networks and IoT, are inherently streaming, involving continuous additions and deletions of vertices and edges at high rates. 
Constructing random walks in a graph, i.e.,~sequences of vertices selected with a specific probability distribution, is a prominent task in many of these graph applications as well as machine learning (ML) on graph-structured data. 
In a streaming scenario, random walks need to constantly keep up with the graph updates to avoid stale walks and thus, performance degradation in the downstream tasks. 
We present \wharf, a system that efficiently stores and updates random walks on streaming graphs. It avoids a potential size explosion by maintaining a compressed, high-throughput, and low-latency data structure.   
It achieves (i)~the succinct representation by coupling compressed purely functional binary trees and pairing functions for storing the walks, and (ii)~efficient walk updates by effectively pruning the walk search space. 
We evaluate \wharf, with real and synthetic graphs, in terms of throughput and latency when updating random walks.
The results show the high superiority of \wharf over inverted index- and tree-based baselines.

\end{abstract}

\maketitle

\pagestyle{\vldbpagestyle}



\ifdefempty{\vldbavailabilityurl}{}{
\vspace{.3cm}
\begingroup\small\noindent\raggedright\textbf{Artifact Availability:}\\
The source code, data, and/or other artifacts have been made available at \url{https://github.com/spapadias/wharf}.
\endgroup 
}

\sloppypar


\section{Introduction}
\label{sec:introduction}
Random walks are used in a large number of graph analysis tasks, such as PageRank~\cite{www99-pagerank, intmath05-perpagerank, vldb10-goel, cikm21-agent}, SimRank~\cite{vldb17-reads}, in influence maximization~\cite{ijcai15-influence, kdd15-influence, kdd10-influence, sigmod18-xiaokui}, in recommendations~\cite{www14companion-rw-in-recsys, kdd15-trustwalk, wsdm22-swalk}, in graph embeddings~\cite{kdd14-deepwalk, kdd16-node2vec}, and graph neural networks~\cite{arxiv20-agl}. For example, random walks-based graph embeddings enable many machine learning (ML) tasks on graphs, e.g.,~link prediction, vertex classification, and outlier detection. 
Thus, computing random walks is at the core of many important tasks today.

Yet, real-world graphs are inherently dynamic, entailing continuous additions and deletions of vertices and edges~\cite{graphChi,arxiv19-vasia,icde15-llama}. 
In many novel applications, such as the Internet of Things and digital twins, graph updates occur with increasingly high frequency, requiring low latency and high throughput processing. For example, Alibaba's e-commerce platform uses massive graphs to store their data~\cite{icde19-keynote-alibaba}:
These graphs consist of billions of vertices (e.g.,~modelling products, buyers, and sellers) and hundred billions of edges (e.g.,~representing clicks, orders, and payments). They use these graphs mainly for link prediction and fraud detection. Alibaba reported these graphs are highly dynamic as they receive a high rate of real-time updates.

\begin{figure}[t]
  \begin{subfigure}[t]{0.23\textwidth}  
    \includegraphics[width=\textwidth]{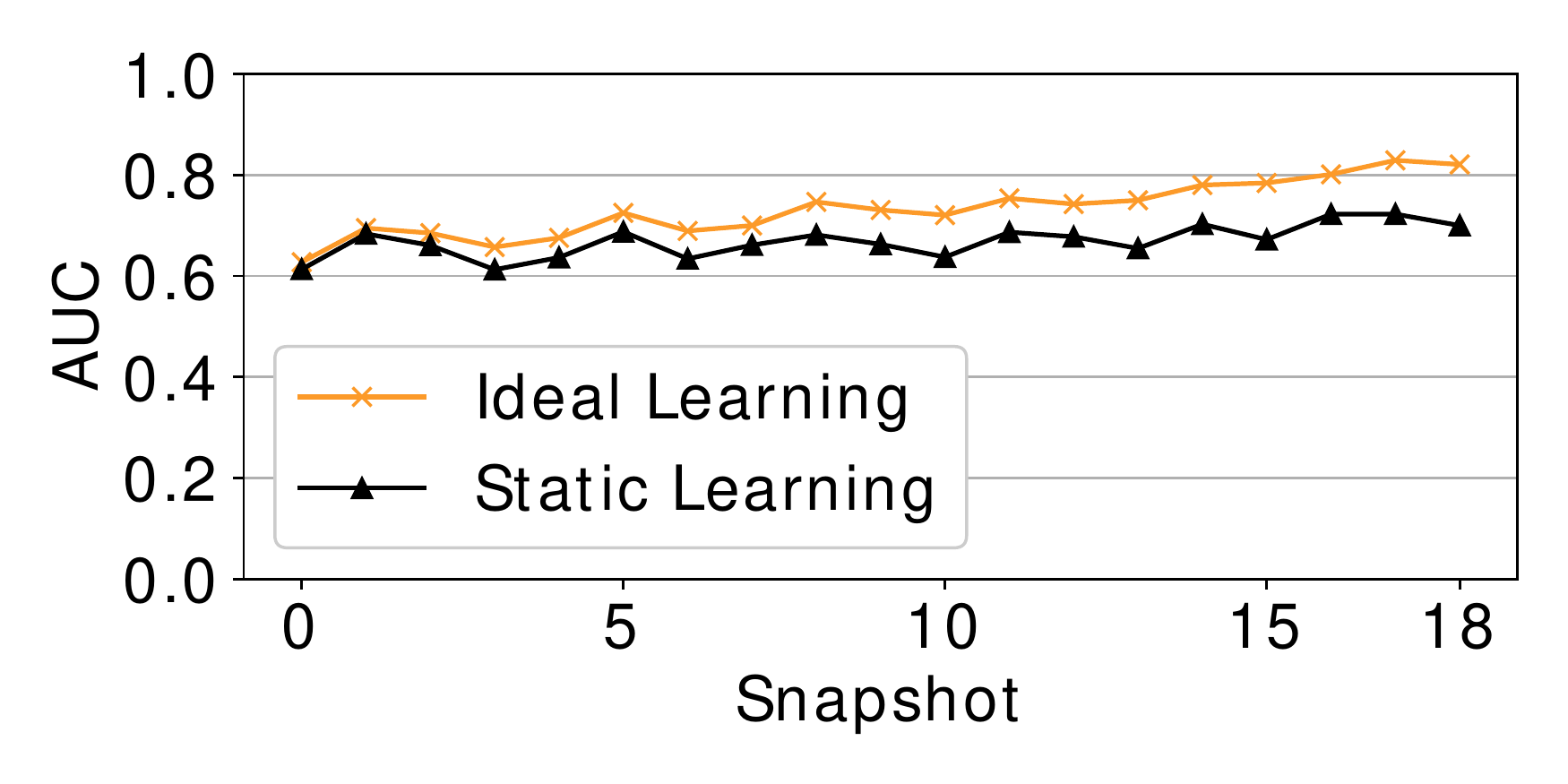}
    \caption{Graph Embeddings}
    \label{fig:motivation-ge}
  \end{subfigure}
    \begin{subfigure}[t]{0.23\textwidth}  
    \includegraphics[width=\textwidth]{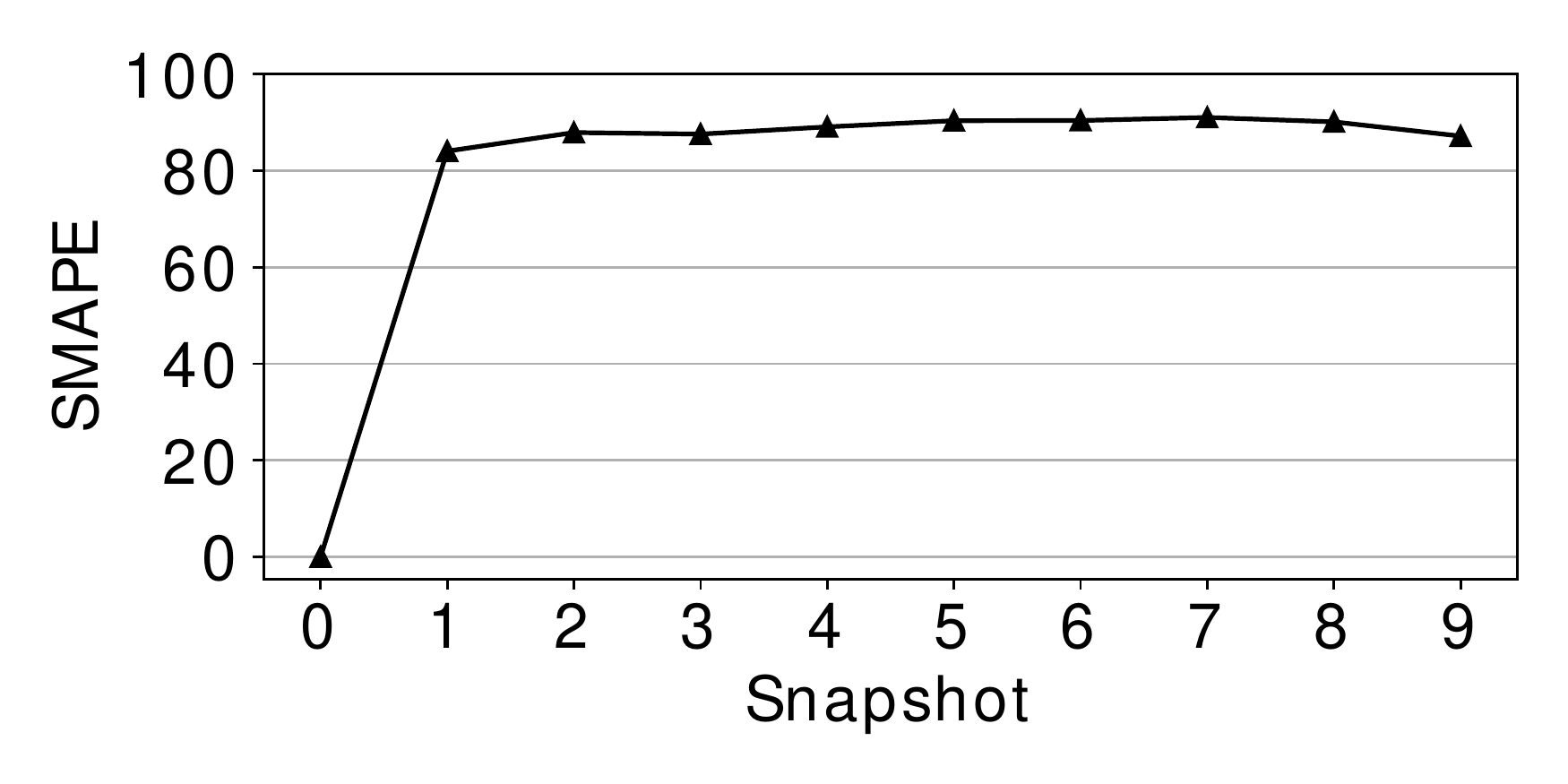}
    \caption{(Personalized) PageRank}  
    \label{fig:motivation-ppr}
  \end{subfigure} 
  \vspace{-0.2cm}
  \caption{Applications of streaming random walks.}
  \label{fig:motivation-introduction} 
  \vspace{-0.4cm}
\end{figure}

Thus, it is crucial to keep random walks up-to-date with the continuous changes to not hurt the effectiveness or accuracy of the downstream tasks. This is exacerbated in high-stake applications, such as anomaly and fraud detection, where even a small percentage of higher accuracy is of utter importance.
Let us illustrate how the accuracy is affected if random walks are not kept up-to-date.
We ran an experiment where we consider a link prediction task on 
a dynamic social network graph after running node2vec~\cite{kdd16-node2vec}, a popular graph embedding algorithm that uses random walks. 
We executed node2vec in two different settings:
static -- we train embeddings only on the initial graph and reuse them subsequently, and ideal -- we retrain embeddings from scratch at each snapshot.\footnote{We provide the detailed setup of this experiment in Section~\ref{subsec:accuracy}.}
Figure~\ref{fig:motivation-ge}
shows the accuracy (i.e.,~AUC score) results.
We observe that retraining embeddings from scratch at each new graph snapshot, i.e.,~after a set of new graph updates have been applied, is mandatory to maintain high accuracy in the downstream ML task (Ideal Learning in the figure).
While in the static scenario, the accuracy drops, in the dynamic scenario the accuracy increases as more graph updates arrive. 
We also ran an experiment where we approximate Personalized PageRank (PPR) scores using~\cite{vldb10-goel} on a dynamic citation graph.  
Figure~\ref{fig:motivation-ppr} shows the Symmetric Mean Absolute Percentage Error (SMAPE) when approximating PPR scores.  
Specifically, we illustrate the SMAPE between the actual algorithm in~\cite{vldb10-goel}, which updates all affected random walks at each snapshot,
and a -- static -- variant, 
which uses the existing random walks.  
We observe that 
the error in PPR scores is above $80\%$ even after the first snapshot arrives. 
We thus expect 
that both the accuracy gap in graph embeddings applications and the estimation error in PPR applications will increase very fast in {\em streaming graphs}, where updates arrive at a very fast rate~\cite{pldi19-aspen,sigmod20-tamer-anil, vldb18-fraud-alibaba, icde19-keynote-alibaba, edbt15-motivation}. 

Despite this importance, the research community has paid little attention to the problem of maintaining up-to-date random walks on streaming (a.k.a,~highly dynamic) graphs.  
We do find a large number of works for efficiently computing random walks~\cite{sosp19-knightking, recsys13-drunkardmob, sigmod20-lei, atc20-graphwalker, icde21-uninet, sc20-csaw}, but all consider static graphs.
Barros et al.~\cite{survey21-barros} present a variety of random walk-based works on Graph Representation Learning (GRL) on dynamic graphs~\cite{csrea19-mitrovic, asonam18-winter, ijcnn19-zhou, ijcai19-tnodeembed, cnta19-evonrl, arxiv19-incremental-node2vec, bigdata18-madhavi, www18-ctdne, cikm20-tdgraphembed}.
Some among those, such as~\cite{arxiv19-incremental-node2vec, cnta19-evonrl}, consider updating random walks but use simplistic inverted indexes to 
do so
and are, thus, inefficient -- they cannot cope with streaming graphs.
There are also theoretical works, such as~\cite{itcs19-jin, soda20-store-a-walk}, that focus on
random walks: \cite{soda20-store-a-walk} studies how to store walks succinctly in an append-only fashion,
which is not applicable for the streaming scenario 
where parts of walks have to be deleted; \cite{itcs19-jin} proposes generating random walks on single-pass graph streams but in an approximate manner.

Updating random walks for streaming graphs is thus an important and open problem. However, doing so is challenging for three main reasons:
(i)~One should update random walks with both low-latency and high-throughput as streams can become quite bursting and volatile with sudden spikes~\cite{icde18-jeyhun}; 
(ii)~One must enable fast access to (parts of) the walks state. 
This allows for realizing fast walk updates by efficiently identifying the walks as well as their parts (vertices) to update;
(iii)~Random walks state should be as succinct as possible especially for applications where the total size of random walks is multiple times larger than the size of the maintained graph, e.g.,~Graph Representation Learning (GRL)~\cite{sigmod20-lei}.

We propose \wharf, a parallel system that tackles all above-mentioned challenges to maintain stateful streaming random walks. 
It stores random walks with the graph within a single data structure forming a hybrid tree-of-trees.
The main idea is to update walks together with the graph:
During a graph update, it identifies the out-of-date walks and updates them in a bulk fashion.

In summary, after giving some preliminaries in Section~\ref{sec:running-example-and-background}, we make the following major contributions: 
\begin{packed_enum}
	\item We formalize the problem of \textit{streaming random walks}, which entails storing walks in a space-efficient way, enabling efficient batch walk updates with high-throughput and low-latency, and fast node retrieval in the set of stored walks simultaneously (Section~\ref{sec:problem-statement}). 
	

\item We propose a novel \textit{hybrid-tree} that stores random walks together with the graph in a compressed form.
 Specifically, we represent random walks as triplets which we encode to integer values using pairing functions. This allows us to
 compact random walks in a lossless manner and maintain the walks state in a way that also serves as an index for efficient walk access.
 Overall, our structure allows for safe parallelism, fast acquisition of lightweight graph and walks snapshots, and high cache locality.
 
	
	
	\item We devise an output-sensitive algorithm for performing efficient search in the set of walks by leveraging the ordering properties of pairing functions (Section~\ref{sec:optimized-search}).
	We also present a walk update mechanism that apply updates in batches and we prove that its time complexity is lower than its competitor
	(Section~\ref{sec:update-random-walks}).
	
	
	\item We validate \wharf through extensive experiments on a variety of real-world and synthetic graph workloads.
	The results show that \wharf achieves its goals in terms of throughput and latency when updating random walks,
	low memory footprint, and effectiveness of the downstream tasks (Section~\ref{sec:experiments}).
\end{packed_enum}

We then discuss related work in Section~\ref{sec:related-work} where we stress that existing non-streaming random walk systems fail to address the above-mentioned challenges. 
We conclude the paper in Section~\ref{sec:conclusion}.

\section{Preliminaries}
\label{sec:running-example-and-background}




We start by providing a running example that we use throughout the paper.
We, then, briefly explain purely-functional trees and the $C$-tree data structure, which our random walk structure leverages, and introduce the pairing functions that are crucial for encoding and porting random walks 
successfully into our random walk structure.

\begin{figure}[t]
	\mbox{
		\begin{subfigure}[t]{0.29 \columnwidth}
			\includegraphics[width=\textwidth]{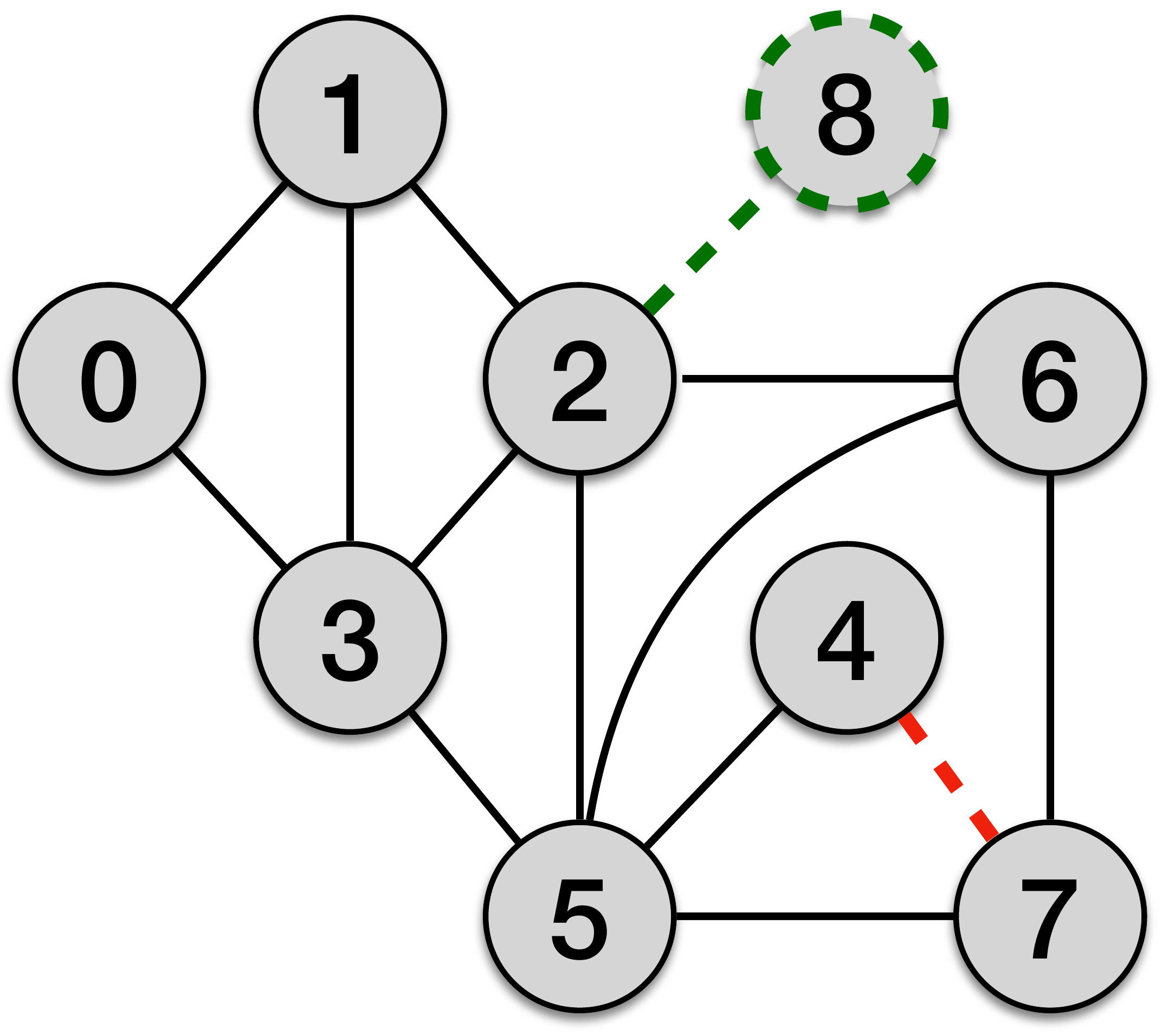}    
			\caption{Example graph}
			\label{fig:ecommerce-graph}
		\end{subfigure}
		\begin{subfigure}[t]{0.36 \columnwidth}
			\includegraphics[width=\textwidth]{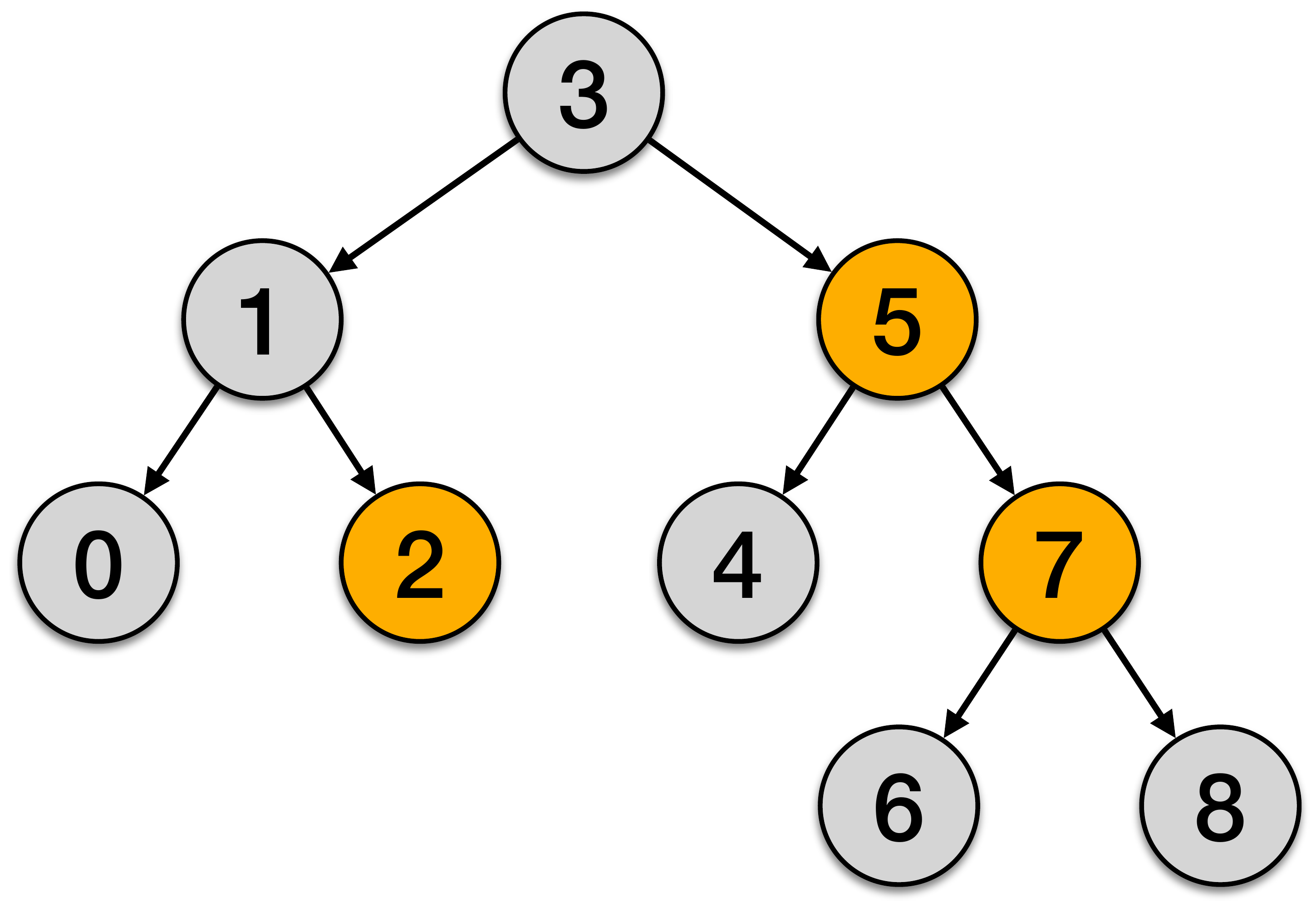}    
			\caption{$PF$-tree}
			\label{fig:pf-tree}
		\end{subfigure}
		\begin{subfigure}[t]{0.32 \columnwidth}
			\includegraphics[width=\textwidth]{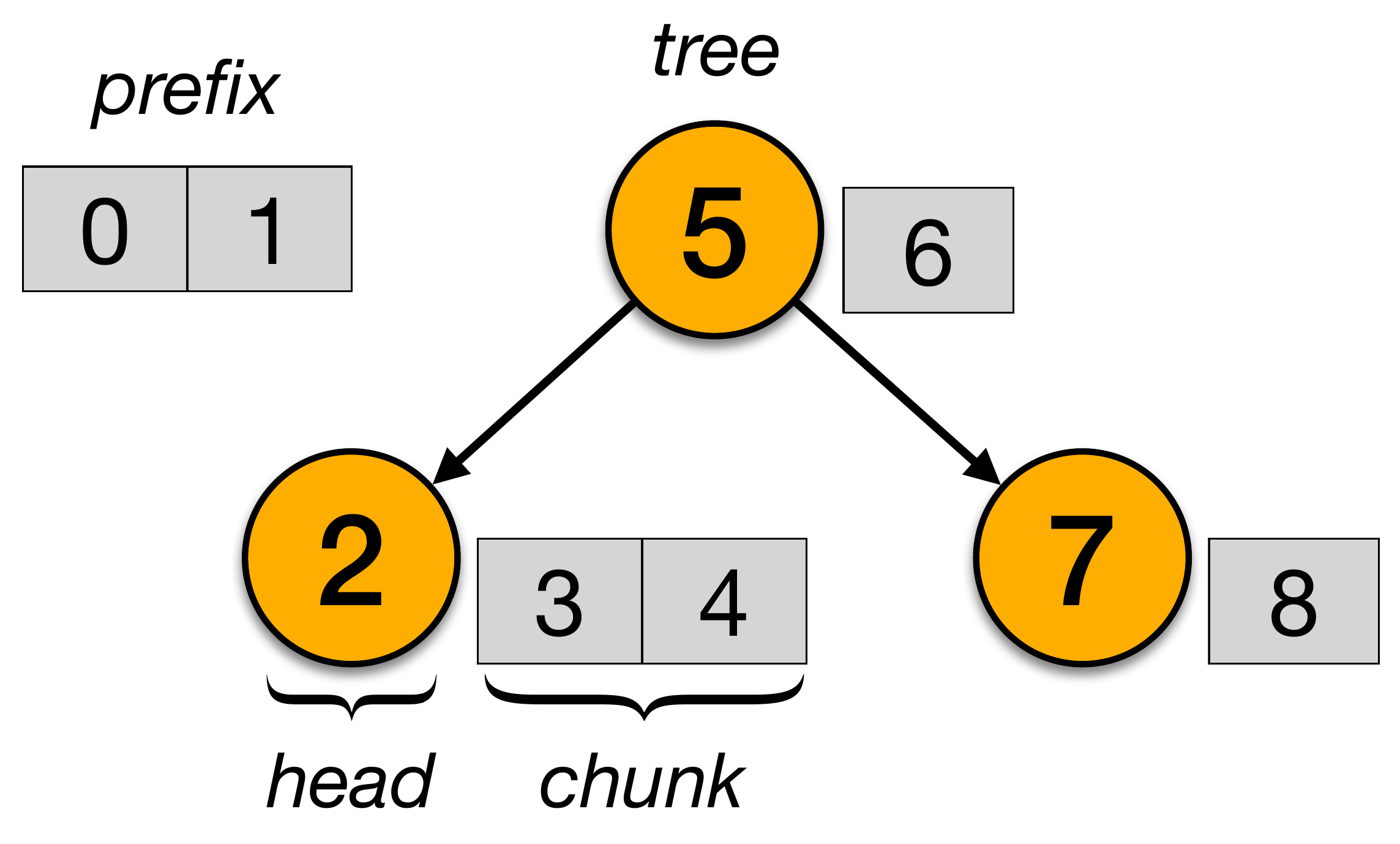}
			\caption{$C$-tree}
			\label{fig:c-tree}
		\end{subfigure}
	}
	\caption{{\bf (a)}~Example graph.
		{\bf (b)}~The corresponding purely-functional binary search tree: Each vertex id is stored in a separate tree node; Orange vertices are heads.
		{\bf (c)}~The corresponding $C$-tree.}
	\label{fig:c(pf)-tree}
\end{figure}

\vspace{0.1cm}
\noindent\textbf{Running example.}
We use an e-commerce graph that contains users and items extracted from the Taobao e-commerce platform~\cite{vldb19-aligraph} as our running example. The vertices correspond to items that can be purchased through the platform, and the edges denote the item-item relationships, i.e.,~items that users purchase together. Figure~\ref{fig:ecommerce-graph}
illustrates an excerpt of the graph.\footnote{Without loss of generality, in this running example we assume a homogeneous graph, yet our techniques apply to heterogeneous graphs as well.} We assign integer identifiers to the vertices for simplicity.

\vspace{0.1cm}
\noindent\textbf{Purely-Functional Trees ($PF$-trees).}
A $PF$-tree is a mutation-free tree structure that preserves 
its former versions when altered and yields a new tree version reflecting the update~\cite{osakaki99-purely}. 
Each element of a $PF$-tree serves as \textit{key}, and  
is 
kept in a separate tree node. 
Figures~\ref{fig:pf-tree}
illustrates the $PF$-tree for our example graph.

\vspace{0.1cm}
\noindent\textbf{Compressed Purely-Functional Trees (\texorpdfstring{$C$}{C}-trees).}
A $C$-tree~\cite{pldi19-aspen} is a binary PF-tree~\cite{osakaki99-purely}, which is additionally compressed and stores multiple elements in each vertex.
A chunking scheme takes the ordered set of elements to be represented and promotes some of them to \textit{heads}, which are stored in a purely-functional tree.
In more detail, given a set $E$ of
elements, one first computes the set of heads $\mathcal{H}(E) = \{e \in E | h(e) \mod b = 0\}$, 
where $b$ is the chunking parameter 
indicating the number of elements each chunk roughly retains, $h : K \rightarrow {1, \dots, N}$ is a hash function drawn from a uniformly random family of hash functions ($N$ is some sufficiently large range).
For each $e \in \mathcal{H}(E)$ let its tail be $t(e) = \{x \in E | \textrm{ } e < x < next(\mathcal{H}(E), e)\}$, where $next(\mathcal{H}(e), e)$ returns the next element in $\mathcal{H}(E)$ greater than $e$. 
Thus, the rest of the elements are stored in tails that are associated with each head vertex of the tree. 
It can exist a headless tail containing the smallest elements to be represented, i.e.,~the \textit{prefix}. 
The prefix, as well as the tails, are both called \textit{chunks}. 
Figure~\ref{fig:c-tree} illustrates the compressed $C$-tree variant of the $PF$-tree in Figure~\ref{fig:pf-tree} for our graph example in Figure~\ref{fig:ecommerce-graph}.
$C$-trees maintain similar asymptotic cost bounds as the uncompressed trees while improving space consumption and cache performance. 
The expected size of chunks in a $C$-tree is $b$,
while the maximum size is w.h.p.~\footnote{The term w.h.p. is an abbreviation of ``with high probability'', i.e.,~with probability $1 - 1/n^c$ for some constant $c$.} $O(b\log n)$. 
The number of heads in a $C$-tree over a set of $n$ elements is $O(n/b)$ w.h.p., and the maximum size of a tail 
or prefix is w.h.p. $O(b\log n)$. 

\vspace{0.1cm}
\noindent\textbf{Pairings.} 
A \textit{pairing function} encodes a pair of natural numbers into a single natural number, uniquely and reversibly.
It is a computable bijection $\pi$ : $\mathbb{N} \times \mathbb{N} \rightarrow \mathbb{N}$. 
We adopt the convention $\langle x,y \rangle$ for a pairing between $x$ and $y$.
Pairing functions share a set of properties with the basic ones being: A pairing function: (i) is an \textit{injection}, (ii) does not contain zero in its range, and (iii) is onto the set $\mathbb{N}^{*}$\footnote{Equivalently, $\mathbb{N}^{*} = \mathbb{N} \setminus \{0\}$, i.e.,~the range of the function does not include zero.}. 
Furthermore, pairing functions possess the following \textit{ordering properties} that are crucial for enhancing the performance of our algorithms,
namely:






\begin{property}[Strict Weak Ordering]
\vspace{-0.2cm}
\label{prop:weak-ordering}
\begin{small}
\begin{align} 
    (\langle x, y \rangle < \langle x', y' \rangle \leftrightarrow x + y < x' + y') \textrm{ or } \nonumber 
    (x + y = x' + y' \textrm{ and } x < x') \nonumber
\end{align}
\end{small}
\end{property}


\begin{corollary}
\label{cor:weak-ordering}
From Property~\ref{prop:weak-ordering} it follows that:
\begin{small}
\begin{equation}
    x + y < x' + y' 
    \rightarrow 
    \langle x, y \rangle \leq \langle x', y' \rangle
\end{equation}
\vspace{-0.6cm}
\end{small}
\end{corollary}


The most well-known pairing functions are Cantor \cite{wiki:pairings} and Szudzik \cite{pdf:szudzik}. 
We adopt the latter one because it ensures that, if both operands are up to $N$-bits, the range of encoded output stays within the limits of a $2N$-bit integer. 
Below, we provide the formulas of $Szudzik(x,y)$ for pairing, and of $Szudzik^{-1}(z)$ for unpairing:


\begin{small}
\begin{flalign} 
    Szudzik(x,y)=
        \begin{cases*}
            y^2 + x & if $x < y$\\
            x^2 + x + y & if $x \geq y$ \nonumber
                 \end{cases*} &&
\end{flalign}
\vspace{-0.2cm}
\end{small}

\begin{small}
\begin{flalign} 
    Szudzik^{-1}(z)=
        \begin{cases*}
            \{z - \lfloor \sqrt{z} \rfloor^2, \lfloor \sqrt{z} \rfloor\} & if $z - \lfloor \sqrt{z} \rfloor^2 < \lfloor \sqrt{z} \rfloor$ \\
            \{\lfloor \sqrt{z} \rfloor, z - \lfloor \sqrt{z} \rfloor^2 - \lfloor \sqrt{z} \rfloor\} & if $z - \lfloor \sqrt{z} \rfloor^2 \geq \lfloor \sqrt{z} \rfloor$ \nonumber
        \end{cases*}&&
\end{flalign}
\end{small}

\section{Problem statement}
\label{sec:problem-statement}
We now formally define the problem we address in this paper.
To do so, we first formalize the streaming graph foundations (Section~\ref{subsec:streaming-graphs}), the notion of random walks and walk corpus (Section~\ref{subsec:random-walks-corpuses}), where we also define statistical indistinguishable random walks. Finally we state the problem of \textit{streaming random walks} (Section~\ref{subsec:stateful-streaming-random-walks}). 


\subsection{Streaming Graphs}
\label{subsec:streaming-graphs}

We assume the popular \textit{edge stream} model, where a graph stream is regarded as a sequence of incoming edges.
We consider \textit{unbounded} graph streams, i.e.,~there is no limit on the number of graph updates that arrive.
A graph update is a set of edge {\em insertions} and {\em deletions}.
In the edge stream model, vertex insertions and deletions happen implicitly via edge updates:
A vertex is added when an edge with a vertex not present in the current graph is inserted, while a vertex is deleted only when its degree becomes zero after an edge deletion.


\label{subsec:graph-state-update}

Following this model, a \textit{streaming graph} is a graph that is subject to edge updates
at a very high rate. 
Specifically, a graph update, $\delta\mathcal{G}$, is a set containing both insertions and deletions of edges. 
Thus, a streaming graph is a long sequence of discrete graph snapshots containing graph updates that should be applied as they arrive.
We formally define a streaming graph as follows:

\begin{definition}[Streaming Graph]
\label{def:dynamic-graph}
A streaming graph is a sequence of discrete graph \textit{snapshots}, $\mathcal{G}^t = \{\mathcal{V}^t, \mathcal{E}^t\}$, 
where $\mathcal{V}^t = \{v_1^t,\dots,v^t_{n}\}$ are the vertices, $\mathcal{E}^t = \{e^t_1,\dots,e^t_{m}\}$ are the edges, and $t \in \mathbb{N}$ is a timestamp.
\end{definition}

\noindent Note that after applying a graph update $\delta\mathcal{G}^t$ 
to a graph snapshot $\mathcal{G}^t$ we end up with the new graph snapshot at timestamp $t+1$, i.e.,~$\mathcal{G}^{t+1}=\mathcal{G}^t+\delta\mathcal{G}^t$.

\subsection{Random Walks}
\label{subsec:random-walks-corpuses}
A random walk on a graph serves as a sample of the graph and is utilized by many applications, such as PageRank~\cite{www99-pagerank, intmath05-perpagerank} and online influence maximization~\cite{ijcai15-influence, kdd15-influence, kdd10-influence, sigmod18-xiaokui}.



\begin{definition}[Random Walk]
\label{def:random-walk}
A 
random walk 
is a $k$-th order Markov chain, where the state space of which is the set of graph vertices $\mathcal{V}$ and the future state depends on the last $k$ steps. 
A random walk $w$ of length $l$ comprises a sequence of vertices, $v_{ 1},v_{2},\dots,v_{j},\dots,v_{l}$, where $v_{j}$ is the $j$-th vertex in $w$ and $j \in \{1,\dots,l\}$, and every two consecutive vertices are connected with an edge.
\end{definition}

In the general case, a random walk $w$ is generated by sampling a vertex $v_j$ given the $k$ previous vertices $v_{j - k}$, $\dots$, $v_{j - 1}$ in $w$ from the following transition probability distribution: $prob(v_j|v_{j-1},\dots,v_{j-k})$.
Note that this probability is non-zero only if an edge between vertices $v_{i-1}$ and $v_i$ exists.
We compute the transition probability following a random walk model, e.g.,~DeepWalk~\cite{kdd14-deepwalk} 
which is a first-order random walk.
In DeepWalk, a walker, which is currently residing at a vertex, consults solely its neighbours to derive the transition probability to select the next vertex for a random walk to visit.
For instance, the walker producing walk $w_0$ (see Figure~\ref{fig:walk-corpus}) moves from vertex $v_1$ to $v_3$ with probability $\frac{1}{3}$, as $v_1$ has three neighbours, namely $v_0$, $v_2$, and $v_3$ (Figure~\ref{fig:ecommerce-graph}). 

\begin{figure}[t]
\includegraphics[width=0.7\linewidth]{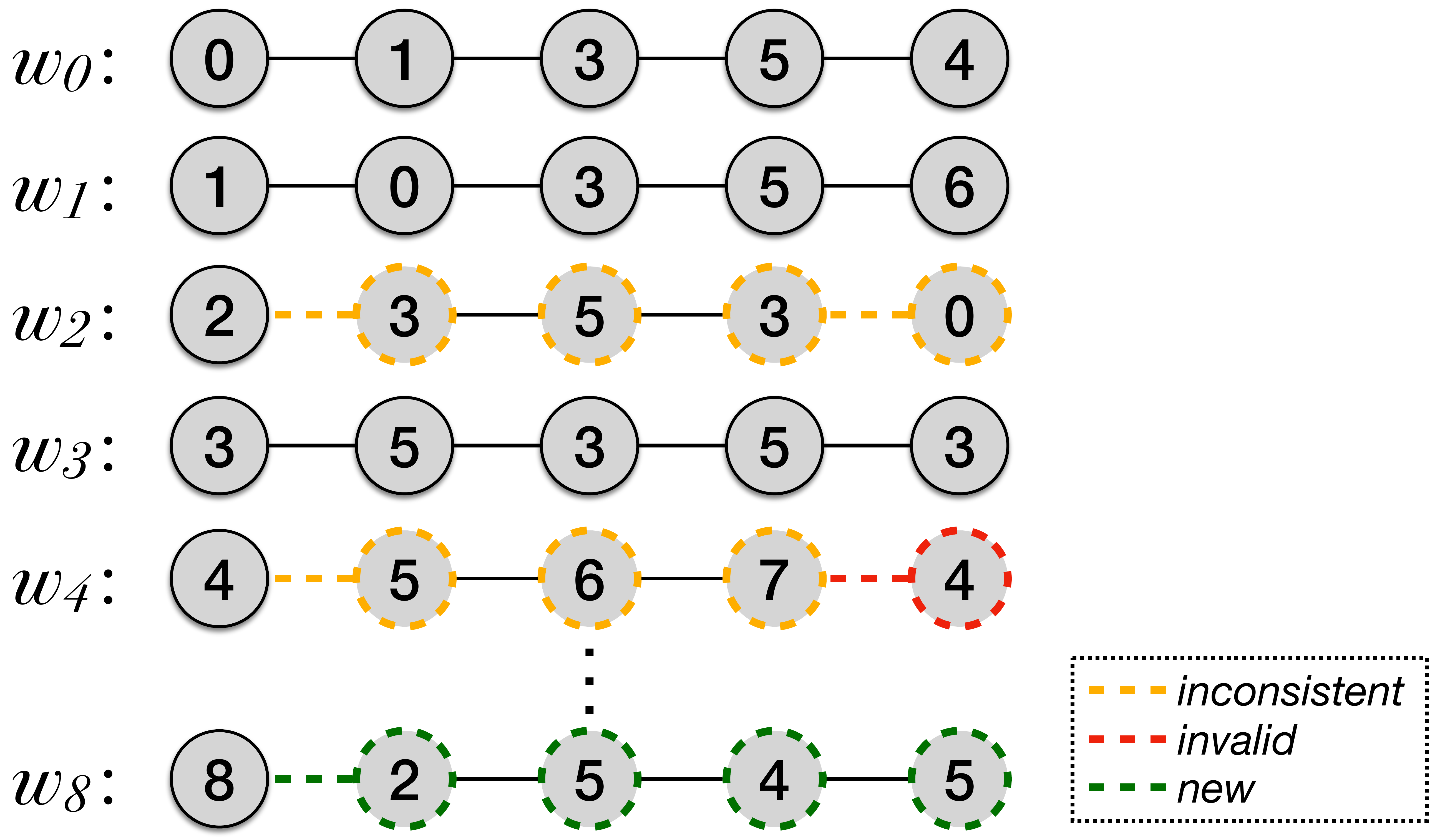}
\vspace{-0.2cm}
\caption{Excerpt of a walk corpus of our graph example.}
\label{fig:walk-corpus}


\vspace{-0.4cm}
\end{figure}

The set of random walks extracted from a graph is referred to as \textit{walk corpus}.
Figure~\ref{fig:walk-corpus} shows an excerpt of a walk corpus that contains a set of random walks for $n_w=1$ and $l=5$, where $n_w$ is the number of walks that initiate from each vertex and $l$ is the length of each walk in $\mathcal{W}$.
As the graph evolves fast, 
walks can become \textit{inconsistent}, and even in certain cases \textit{invalid}:
A random walk is inconsistent when it
does not reflect the graph transition probabilities correctly;
An invalid random walk, in contrast, is an inconsistent walk that has been disrupted by an edge deletion and it cannot be recreated in the updated graph.
Note that in the sequel, whenever it is not necessary to differentiate between inconsistent and invalid walks we refer to them simply as \textit{affected}.

We illustrate both inconsistent and invalid random walks in the walk corpus of Figure~\ref{fig:walk-corpus}. 
Assume that the edge deletion of $\{4, 7\}$ happens before the edge addition of $\{2, 8\}$ in the example graph of Figure~\ref{fig:ecommerce-graph}. In this case, $w_2$ becomes inconsistent as the transition probabilities for a walker residing on graph vertex $v_2$ change and thus we need to refine all subsequent walk vertices of $w_2$. 
Also, $w_4$ becomes invalid due to the edge deletion, and at the same time inconsistent as the transition probabilities of vertex $v_4$ change. 
In our work, we choose to refine both invalid and inconsistent walks, to ensure \textit{statistical indistinguishability} as we describe next. 

Ideally, we want to update only those random walks that become inconsistent or invalid after a graph update.
Updating one of these walks
means updating all its {\em affected vertices}, i.e.,~those vertices that make the random walk inconsistent or invalid.
This is because they might not match the transition probabilities of the updated graph $\mathcal{G}^{t+1}$.
For instance, $w_2$ becomes inconsistent and we need to refine all subsequent walk vertices of walk $w_2$, e.g.,~producing an updated walk $v_2, v_3, v_5, v_4, v_2$.
Note that random walks that do not match the transition probabilities of the graph can lead to low accuracy of downstream tasks.
We say that a walk corpus is up-to-date when it is {\em statistically indistinguishable}~\cite{arxiv19-incremental-node2vec}, which we 
define as follows:

\begin{property}[Statistical Indistinguishability]
\label{def:statistical-indistinguishability}
A walk corpus $\mathcal{W'}$, resulting from updating a walk corpus $\mathcal{W}$ after a graph update $\mathcal{\delta\mathcal{G}}$, is statistically indistinguishable if 
it is equi-probable with a new walk corpus 
that is generated from scratch on graph $\mathcal{G}' = \mathcal{G} + \delta\mathcal{G}$.
\end{property}

\subsection{Streaming Random Walks}
\label{subsec:stateful-streaming-random-walks}

We now formally define
the problem of computing \textit{streaming random walks}. 
Specifically, we aim at representing walk corpuses space-efficiently, enabling incremental updates, and allowing for fast walks access. 
In the sequel, we refer to a specific random walk algorithm as {\em random walk model}.
Formally:





\begin{problem}
Given a random walk model $\mathcal{M}$, we define the problem of streaming random walks as (i)~maintaining and storing in main memory a walk corpus $\mathcal{W}$ that is generated based on $\mathcal{M}$, (ii)~ensuring that $\mathcal{W}$ is always statistically indistinguishable, and (iii)~updating $\mathcal{W}$ incrementally without recomputing it from scratch.
\end{problem}

Tackling the above problem 
is challenging for three main reasons:
(1)~We should avoid external structures to index the walks to avoid expensive index updates
(Section~\ref{sec:graph-walk-structure});
(2)~We should efficiently navigate through random walks (Section~\ref{sec:optimized-search}); 
(3)~We should efficiently update random walks while they are in their compressed form to avoid time-consuming (de)compression processes 
(Section~\ref{sec:update-random-walks}).

\section{Graph-Walk Structure} 
\label{sec:graph-walk-structure}
 
Our goal is to come up with a data structure that stores the random walk corpus together with the streaming graph. In this way, we can both speed up the update process and reduce the required storage space. In addition, we aim at indexing the random walks so that we can quickly locate the (parts of) random walks that require updating and keep them up-to-date with the graph updates. One may think that maintaining a simple inverted index for the walks suffices for fast access, yet it is not efficient for updating walks in a streaming scenario as we will show in the experimental evaluation.

In a nutshell, we propose a \textit{hybrid-tree} data structure  (Section~\ref{subsec:hybrid-tree-wharf}) that aims at tackling both the challenge of 
efficiency and space.
The hybrid-tree enables not only efficient graph updates, but also efficient access to random walks, and consequently, efficient random walk updates by avoiding calculating all walks from scratch.
We store random walks as a set of triplets within the hybrid-tree (Section~\ref{subsec:walk-triplets}).
Such a triplet representation serves at the same time as an index on the random walks.
We use pairing functions with ordering properties to encode the derived triplets into integers and reduce space.
(Section~\ref{subsec:pairing-triplets}).
This allows us to further compress the random walks using difference encoding 
(Section~\ref{subsec:difference-encoding}).

\subsection{Hybrid Tree}
\label{subsec:hybrid-tree-wharf}

We illustrate the hybrid-tree data structure in Figure~\ref{fig:wharf-representation}. 
The hybrid tree is a tree-of-trees where each node of the outer tree consists of an id and two trees.
All the vertices of the graph are stored in the outer tree, which we call \textbf{\textit{vertex-tree}} (outer gray
nodes in Figure~\ref{fig:wharf-representation}).
Each vertex in the vertex-tree (outer tree) stores the identifiers of its adjacent neighbours in a \textit{C}-tree, which we call \textbf{\textit{edge-tree}} (inner left blue tree).
In addition, each outer vertex also stores the parts (vertices) of the random walks in which it participates in a second $C$-tree, which we call \textbf{\textit{walk-tree}} (inner right yellow tree).
This design 
allows us to efficiently access specific random walks in a walk corpus
and update them together with a graph update.
Specifically, a hybrid-tree enables both fast access to a specific entry of a vertex in a random walk and storing large walk corpuses.

\begin{figure}[t]
\centering
\includegraphics[width=0.45 \textwidth]{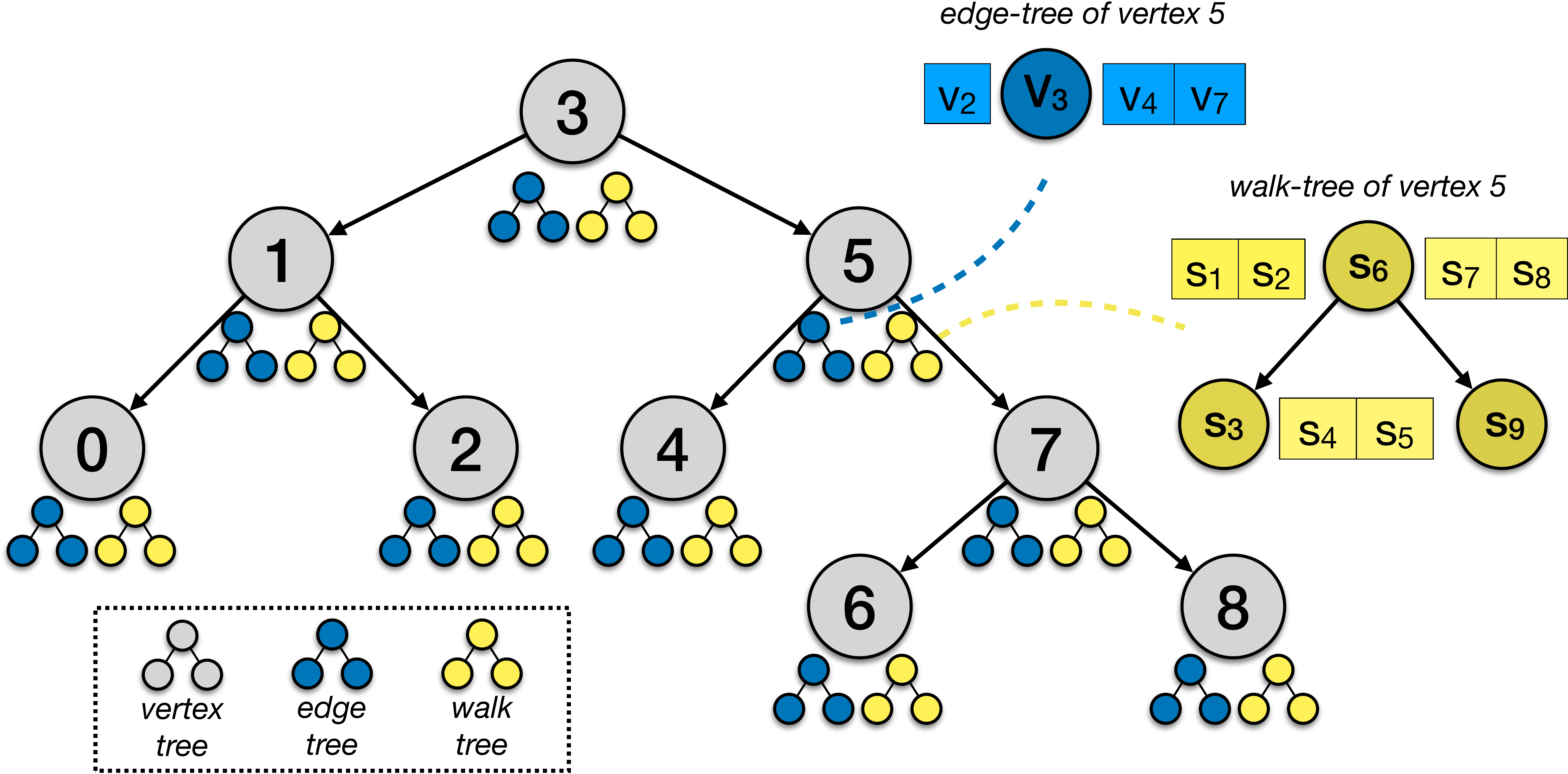}
\vspace{-0.2cm}
\caption{\wharf's \textbf{\textit{hybrid-tree}}. The edge-tree of vertex $5$ contains its neighbors ($v_2$, $v_3$, $v_4$, $v_7$) as illustrated in Figure~\ref{fig:ecommerce-graph} and the walk-tree of vertex $5$ contains the encoded walk triplets ($s_1,\dots,s_9$) that correspond to $v_5$'s entries in the walk corpus of Figure~\ref{fig:walk-corpus}. 
} 
\label{fig:wharf-representation}
\vspace{-0.2cm}
\end{figure}

The hybrid-tree is a two-level tree structure that has $O(\log n)$ overall depth using any balanced binary tree implementation. 
For the vertex-tree, as well as for storing the heads of the edge-trees and walk-trees, we use a {\em parallel augmented map} (PAM)~\cite{ppopp18-pam}, as they enable safe parallelism and lightweight snapshots. 
In our design, the edge- and walk-trees are subcomponents 
to ensure that once we acquire a snapshot of the purely-functional vertex-tree, we directly gain access to the state of both the graph and the walk corpus. 
Note that we represent the edge- and walk-trees with the $C$-tree structure as it allows for safe parallelism, lightweight snapshots, strict query serializability, efficient space usage, and cache locality.
\subsection{Walk Triplet Representation}
\label{subsec:walk-triplets}

We utilize a triplet-based representation to store a random walk within the walk-tree: 
For each vertex of the graph, we keep the walk id that the vertex participates in, the position of the vertex in the walk, and the next vertex of the random walk. 

In detail, each vertex $v$
in a corpus $\mathcal{W}$ can be uniquely described by the pair $(w_{i}, p_{j})$, where $w_{i}$ is the walk it participates in, with $i \in \{1, ..., |\mathcal{W}|\}$, 
and $p_{j}$ is the position of $v$ 
in $w_{i}$ where $j \in \{1,...,l\}$. 
In other words, a pair $(w_{i}, p_{j})$ serves as the coordinates of the vertex in the corpus $\mathcal{W}$ and thus, when seeking for $v$, this pair behaves as its search key. 
We, thus, denote with $v_{w_i,p_j}$ the identifier of a vertex $v \in \mathcal{W}$.
Note that a vertex may appear multiple times in a walk corpus as well as in the same walk. 
Instead of storing raw walk sequences, we group the walk triplets by vertex identifier and store them in their corresponding walk-trees of the hybrid-tree structure 
associated with 
the appropriate vertex of the vertex-tree.
\emph{This enables fast access
to specific vertices of walks in the corpus and space efficiency.}
However, instantly accessing an affected vertex in the corpus is not enough.
We should also be able to traverse a walk so that we can update it efficiently.
To achieve this, we maintain the vertex identifier of the next node, $v_{w_{i}, p_{j+1}}$ at position $p_{j+1}$ in a walk $w$ as the third element of a walk triplet.
We thus represent each vertex $v_{w_{i}, p_{j}}$ with a \textit{walk triplet} of the form $(w_{i}, p_{j}, v_{w_{i}, p_{j+1}})$. 
Ultimately, \emph{these walk triplets serve for both storing and traversing any random walk sequence efficiently}. 
For example, in the corpus of Figure~\ref{fig:walk-corpus}, walk $w_0$ can be represented as a sequence of the following triplets: $(w_0, p_0, v_1), (w_0, p_1, v_3), (w_0, p_2, v_5), (w_0, p_3, v_4), (w_0, p_4, v_4)$. 
Note that the position numbering starts from $0$ and that for the last vertex of a walk, the next vertex of its walk triplet has the same vertex id with the vertex itself denoting the end of the walk\footnote{Note that we can use any other termination identifier such as the integer $-1$.}. 


\subsection{Walk Triplet Pairing}
\label{subsec:pairing-triplets}
Storing integer values instead of entire triplets objects achieves great space savings. This is not only because of the object footprint but also because it allows for difference encoding.
We, thus, convert the walk triplets into integers to store them in the $C$-trees. 
Our main idea is to encode each walk triplet into a unique integer
via a pairing function. 
We could easily achieve this encoding
by two invocations of a pairing function: encoding the first two elements of the triplet into a paired value and then encode again the paired value with the third element.
Yet, pairing comes at a cost. 
Recall from Section~\ref{sec:running-example-and-background} that for two operands that are up to $N$ bits, Szudzik pairing function returns a $2N$-bit number.
Thus, the fewer pairing function invocations, the smaller the encoded walk triplet values.

We, thus, reduce the number of pairing function invocations by first encoding the walk identifier of a vertex along with its position in the corresponding walk together into a single number.
Specifically, for a walk triplet $(w_{i}, p_{j}, v_{w_{i}, p_{j+1}})$ and given that the length of $w_{i}$ is $l$, we devise the following function to encode $w_{i}$ and $p_{j}$ into a single integer: $f(w_i, p_j)=w_{i}\times l + p_{j}$.
We, then, invoke Szudzik once to pair the output of this function, $f(w_i, p_j)$,  with the identifier of the next node in the walk, $v_{w_{i}p_{j+1}}$: $\langle f(w_i, p_j), v_{w_{i}, p_{j+1}}\rangle$.
When we unpair an encoded walk triplet, we retrieve the walk id and position from $f$ as follows: $w_{i} = \left \lfloor{\frac{f}{l}}\right \rfloor \text{~and~}  p_{j} = f \;\mathrm{mod}\; l$.



\noindent Note that $p$ is upper bounded by $l$, and thus, we can utilize the simple function $f$ to encode $w$ and $p$ as well as revert to the original values with the above-mentioned equations. However, 
in the streaming setting there is no upper bound for $v_{w_{i}, p_{j+1}}$, so we rely on a pairing function for the final encoding.  
Specifically, we used the Szudzik function because it
ensures that for two $N$-bit integer arguments, its value is at most a $2N$-bit integer, and thus, guarantees that there will be no integer overflows.
For instance, in our running example for triplet $(w_0, p_0, v_1)$, we invoke $Szudzik\langle w_0 \times l + p_0, v_1 \rangle$
to get the integer value that we will insert in the $C$-tree.
Thus, function $f$, as well as $v_{w, p_{j+1}}$, 
must be at most $N$-bit numbers.
Formally: 
%

\vspace{-0.2cm}
\begin{small}
\begin{equation}
    f(w, p) = w\times l + p \leq 2^N-1 \wedge \nonumber 
    v_{w, p_{j+1}} \leq 2^N-1, \quad where\;p_{j+1} \leq l
    \label{eq:pairings-range-limitation}
\end{equation}
\end{small}
\vspace{-0.4cm}

\noindent which dictates the cap of maximum values for $w$, $l$, and $v_{w, p_{j+1}}$. 
Encoded 
triplets 
go in the walk-tree of the 
vertex they correspond. 
Let us now illustrate how the walk- and edge-trees in our running example are populated.
Focusing on vertex $v_5$, assume that it only appears in the walks that are shown in Figure~\ref{fig:walk-corpus} as well as in the first position of $w_5$ with $v_7$ as its next vertex (not shown). 
Figure~\ref{fig:wharf-representation} shows the contents of $v_5$'s walk-tree. 
The walk triplets of vertex $v_5$ are: $(w_0, p_3, v_4)$, $(w_1, p_3, v_6)$, $(w_2, p_2, v_3)$, $(w_3, p_1, v_3)$, $(w_3, p_3, v_3)$, $(w_4, p_1, v_6)$, $(w_5, p_0, v_7)$, $(w_8, p_2, v_4)$, $(w_8, p_4, v_5)$. After the encoding we get the integer values $s_1, s_2, \dots, s_9$, respectively, where $s_1 < \dots < s_9$ holds without loss of generality (w.l.g.)
As we see in Figure~\ref{fig:wharf-representation}, \wharf stores the encoded triplets 
monotonically inside the walk-tree and $s_3, s_6, s_9$ are selected as head vertices. 
In addition, Figure~\ref{fig:wharf-representation} shows the edge-tree of $v_5$ that contains its neighbours in our running example graph (Figure~\ref{fig:ecommerce-graph}), which are $v_2, v_3, v_4, v_7$. 
Assuming that $v_2 < v_3 < v_4 < v_7$ (w.l.g.), they are monotonically stored inside $v_5$'s edge-tree ($v_3$ acts as head). 

\subsection{Walk Triplet Compression}
\label{subsec:difference-encoding}

It is worth noting that pairing invocations produce 
numbers that are much larger than their arguments, 
which incurs a large storage overhead.
Difference encoding (DE) alleviates this problem, as we store only the differences of the integers that correspond to encoded walk triplets in each chunk. 
More specifically, we exploit the fact that trees store elements (integer values) in sorted order in chunks to further compress the data structure.
Given a chunk containing $d$ integers, $\{I_1, I_2,\dots, I_d \}$, we compute the differences $\{I_1, I_2 - I_1,\dots, I_d - I_{d-1}\}$ and encode them 
using a variable byte-code \cite{dcc15-ligra+}. 
Note that after encoding the walk triplets with the Szudzik, we store them monotonically in increasing order in the $C$-trees, and thus, the differences produced by the 
DE are always non-negative.
Clearly, the difference encoding scheme applied to the chunks counterbalances the fact that we store ``big'' numbers produced by the pairings. 
Note that each chunk must be processed sequentially, namely decompressed and then re-compressed as a whole. 
The cost of the sequential decoding does not affect the overall work or depth of parallel tree methods, as the size of each chunk is small ($O(\log n)$ w.h.p.) for a constant chunking parameter $b$. 
Similarly, chunks must be re-compressed when receiving updates, which has cost on par with the cost of decompressing the chunks.   
We consider more sophisticated encoding schemes that avoid decompressing and re-compressing as future work.

\subsection{Space Complexity}
\label{subsec:space-complexity-comparison}

Let us now elaborate on the memory footprint that \wharf needs to store the walks. 
Assume we use $B$-bit integers, \wharf needs $B$ bits for each encoded walk triplet. 
The total number of walk triplets in a walk corpus is $|W| = n * n_w * l$, where $n$ is the number of vertices in the graph, $n_w$ is the walks per vertex, and $l$ is the length of each walk.  
Therefore, \wharf needs $\Theta(|W| \times B)$ space to store the walks. 
This is because we have one encoded walk triplet for each vertex in the walk corpus. 
On the other hand, a simplistic inverted index-based solution similar to~\cite{cnta19-evonrl} that stores the whole walk corpus sequentially needs $\Theta(|W| \times B)$ space for the walk sequences. 
Additionally, for the inverted index that relates each vertex id with the set of walk ids it participates, it needs $O(2 \times |W| \times B)$ space. 
Thus, total space complexity ends up being $O(3 \times |W| \times B)$. 

\section{Optimized Search}
\label{sec:optimized-search}

Before delving into the details of how we update random walks, we first discuss one of the 
factors that make \wharf highly performant in updating random walks: its capability to search in walk-trees so that walk traversal is possible.
Traversing walk-trees is challenging for two reasons. 
First, 
a random walk is 
represented as a set of triplets stored under different vertices of the vertex-tree.
Second, the only available information in a walk-tree is unique integer values, which are the encoded triplets.
In particular, given a vertex $v$ of a random walk $w$, the operation for finding the next vertex is essentially searching for a triplet $(w, p, *)$, but without knowing the actual value of its third element.
Thus, when seeking the next vertex in a walk of the corpus, the pair $\{w, p\}$ serves as a search key. 
The fact that walk-trees are filled with integer values representing encoded triplets, 
prevents us from directly using the search key to find the triplet we are looking for.
A trivial way of finding it would be to visit each vertex of the walk-tree, decode its encoded triplets to retrieve the original ones, and check if one of them corresponds to walk $w$ at position $p$.
In the worst case, we would decode all elements in the tree even if the triplet does not exist. 
The complexity of this process is $O(n)$ where $n$ is the number of elements in 
a walk-tree, which is prohibitive for large-scale streaming graphs.
Next, we describe an efficient 
search algorithm based on range queries. 

\small
\begin{algorithm} [t]
	\caption{\textsc{FindNext}}
	\label{alg:find-in-range}
	
	\begin{algorithmic}[1]
	    \State \textbf{Input:} walk-tree $WT$, walk id $w$, position $p$
	    
	    \State \textbf{Output:} next vertex $v_{w, p+1}$

		\State{$lb=\langle w \times l + p, WT.v^{min}_{w, p+1} \rangle$}
	    \Comment{lower bound search range}
	    
        \State{$ub = \langle w \times l + p, WT.v^{max}_{w, p+1} \rangle$}
	    \Comment{upper bound search range}
	    
        
			\If{$WT$ is $Empty$}

	            
	            \Return null
	        

			\ElsIf{$WT.prefix$ is $Empty$}

    	    \Return{\textsc{TraverseTree}($WT.tree.root, w, p, lb, ub$)}
    	    

			\Else
    	        
    	        \If {$ub \geq WT.prefix.first$ \textbf{or} $lb \geq WT.prefix.last$}
    	        
    	        \Return{\textsc{ExamineChunk}($WT.prefix$)}
    	         
                \Else
                
                \Return{\textsc{TraverseTree}($WT.tree.root, w, p, lb, ub$)}
    	          
    	        \EndIf
    	     \EndIf
    
	\end{algorithmic} 
\end{algorithm}
\normalsize

\subsection{Search Space Pruning}
We start by describing how we enable efficient searching over walk-trees without necessarily decoding all walk triplets in the worst case. 
The use of pairing functions is a calculated move, as they have properties that enforce ordering among triplets of a walk-tree.
We leverage this ordering property to reduce the search space in a walk-tree and hence achieve efficient search.

Based on Corollary~\ref{cor:weak-ordering}, we can construct a \textit{search range} $[lb, ub]$ for a vertex of walk $w$ at position $p$ where: 
 
\vspace{-0.1cm}
\begin{small} 
\begin{equation}
    lb = \langle w \times l + p, v^{min}_{w, p+1} \rangle \nonumber 
    ~\text{~and~}~
    ub =  \langle w \times l + p, v^{max}_{w, p+1} \rangle \nonumber
\end{equation}
\end{small}
\vspace{-0.4cm}

\noindent $v^{min}_{w, p+1}$ and the $v^{max}_{w, p+1}$ are the minimum and maximum next 
 vertex ids that appear in all the walk triplets of the walk-tree, respectively. 
We calculate the pair $\{v^{min}_{w, p+1}, v^{max}_{w, p+1}\}$ at the time we construct a tree and refine it when we update the random walks.
Conceptually, the search range is a subset of the range $[min, max]$ where:

\vspace{-0.1cm}
\begin{small}
\begin{align}
    min = \langle w_{min} \times l + p_{min}, v^{min}_{w, p+1} \rangle \nonumber 
        \text{~and~}
    max = \langle w_{max} \times l + p_{max}, v^{max}_{w, p+1} \rangle \nonumber
\end{align}
\end{small}
\vspace{-0.4cm}

\noindent The $min$ and the $max$ are the global minimum and global maximum encoded values that can be possibly found in a walk-tree.
Consequently, the range $[min, max]$ encloses all the walk triplets inside the tree.
$[lb, ub] \subseteq [min, max]$ holds for the two aforementioned ranges.
%
%
Therefore, if the walk triplet exists in the walk-tree, its encoded value must exist inside the range $[lb, ub]$.

\subsection{Next Vertex Search}
Algorithm \ref{alg:find-in-range} illustrates the \textsc{FindNext} operation which intuitively performs a range query in the reduced \textit{search range} as defined above. 
The algorithm receives as input a walk-tree $WT$,
a walk identifier $w$, and a position $p$, and returns the vertex $v_{w, p+1}$ at position $p+1$ of walk $w$. 
We initiate our search from the prefix part of the walk-tree (Lines~3-4). 
If the triplet is not found inside the prefix, then we continue the search in the tree part of the walk-tree. 
Note that Algorithm~\ref{alg:find-in-range} calls the \textsc{TraverseTree}($root, w, p, lb, ub$) procedure (Lines~6 and~9), 
which recursively traverses the tree part of a walk-tree while searching for matching walk triplets.
As pointed out in \cite{pldi19-aspen}, it is quite important to efficiently compute the first and last elements of a chunk $c$, i.e.,~the $c_{first}$ and $c_{last}$, respectively. 
Recall a chunk stores the encoded triplets.
Of course, $c_{first} < c_{last}$ holds. To avoid scanning whole chunks,
the first and last elements are stored at the head of each chunk for fetching $c_{first}$ and $c_{last}$ in $O(1)$ work and depth.
This modification is important to ensure that \textsc{FindNext} can be done in $O(b\log n + k)$ work and depth w.h.p. on a walk-tree, where $k$ is the number of encoded triplet values lying within this search range of $WT$. 
We can skip searching in a chunk $c$ (either in the prefix or in the tree) if $ub < c_{first}$ or $lb > c_{last}$, because all the encoded triplets inside $c$ are outside the search range (Line~9).
It is worth noting that as the tree part of $WT$ is actually a binary search tree and its encoded triplets are stored in increasing order, we search it by conducting an \textit{in-order} traversal. 

\subsection{Complexity}
We now discuss the complexity of our optimized search algorithm for finding a walk triplet at position $p$ of walk $w$ in a walk-tree $WT$, and in a search range $[lb,ub]$ of triplets encoded via a (constant work) pairing function. 
We conduct two root-to-leaf path searches 
based on the [$lb, ub$] 
range, which have complexity $O(b\log n)$.  
The 
range essentially dictates which ``internal'' walk-tree nodes that are enclosed in these two paths to search exhaustively. 
Then, assuming there are $k = |\{e | e \in $WT$ \textrm{ and } e \in [lb,ub]\}|$ leaves between the leaves of the two aforementioned search paths, we have to traverse them all, which has $O(k)$ complexity.  
Finally, the total complexity for the output-sensitive range search is $O(b\log n + k)$.

\section{Updating Random Walks}
\label{sec:update-random-walks}


\wharf applies walk updates in batches and in parallel.
It receives graph additions and deletions and buffers them to apply them in bulk. 
This allows \wharf to perform fast walk updates. 
It identifies the affected vertices, while updating the graph, and updates all those random walks that contain them.
This is important
as the number of affected vertices are several orders
of magnitude smaller than the total number of vertices.
Next, we describe how we compute and update the structure that keeps the affected vertices every time a batch of updates arrives. Then, we present our update algorithm, whose input includes the computed affected vertices structure. 

\subsection{Map of Affected Vertices} 
\label{subsec:man-computation}
We construct a {\em map of affected vertices} ($MAV$), while processing a graph update, to be able to update only the affected vertices.
In a nutshell, the $MAV$ is responsible for bookkeeping 
key affected vertices in each affected walk.
Note that in each affected walk the first encountered affected vertex is of special importance. 
This is because some of (if not all) the transition probabilities with which we sampled the subsequent vertices do not match the new probabilities in the updated graph.
Using walks that do not reflect the graph structure accurately can lead to low accuracy of downstream tasks that rely on them.
We thus formally define the $MAV$ as follows:

\begin{definition}[Map of Affected Vertices -- MAV] A MAV is a key-value map that contains affected vertices for each affected walk in a walk corpus $\mathcal{W}$:
	The key is the identifier of an affected walk $w$ and its value is the pair $\{v_{min}, p_{min}\}$, with $v_{min}$ being the first affected vertex in $w$ located at position $p_{min}$.
\end{definition}

We compute the $MAV$ as follows. 
Assume, without loss of generality, a batch of graph updates, $\delta\mathcal{G}$, containing undirected 
edges: 
with each edge $e = (s, d) \in \delta\mathcal{G}$ being treated as two directed edges, namely, one $e_1$ initiating from a source vertex $s$ to a destination vertex $d$ and another $e_2$ starting from $d$ to $s$.
Once an edge is incorporated into the appropriate edge-trees (one for each direction), it may render an existing walk in the maintained corpus \textit{inconsistent} or even worse \textit{invalid}. 
Specifically, based on the edge update (either insertion or deletion), 
we identify the affected walks and vertices from the walk-trees.
We distinguish the following two cases with respect to an updated edge
$e_1$
(similarly for the other direction $e_2$):

\noindent{\em (1) Edge Insertion}: After the insertion of an edge $e_1$, any walk $w\in\mathcal{W}$ containing vertex $s$ becomes inconsistent because its transition probability is not the same anymore;
In this case, we insert $(w, \{s, p_s\})$, where $p_s$ is the position of $s$ in $w$, into the $MAV$ if an entry for $w$ does not exist, otherwise, we update its entry with the pair $\{s, p_s\}$, if $p_s$ is smaller than the current $p_{min}$. 

\noindent{\em (2) Edge Deletion}: After the deletion of an edge $e_1$, any walk $w\in\mathcal{W}$ containing vertex $s$ becomes inconsistent, but it is invalid if it also contains a transition from $s$ to $d$; 
We update the $MAV$ exactly as in the case of edge insertion.
Our hybrid-tree allows us to efficiently update the $MAV$, as we only have to search the walk-tree of the source vertex that belongs to an edge addition/deletion.

\subsection{Batch Walk Update}
\label{subsec:batch-walk-updates}


The main idea is to translate a batch of graph updates into a batch of walk updates.
We do so by populating 
an \textit{insertion accumulator} 
which gathers the encoded triplets that correspond to the newly sampled vertices and then bulk-insert them in the corresponding walk trees. 
A \textit{merge} process then evicts the obsolete walk triplets.

\begin{figure}[t]
	\centering
        \includegraphics[width=0.35 \textwidth]{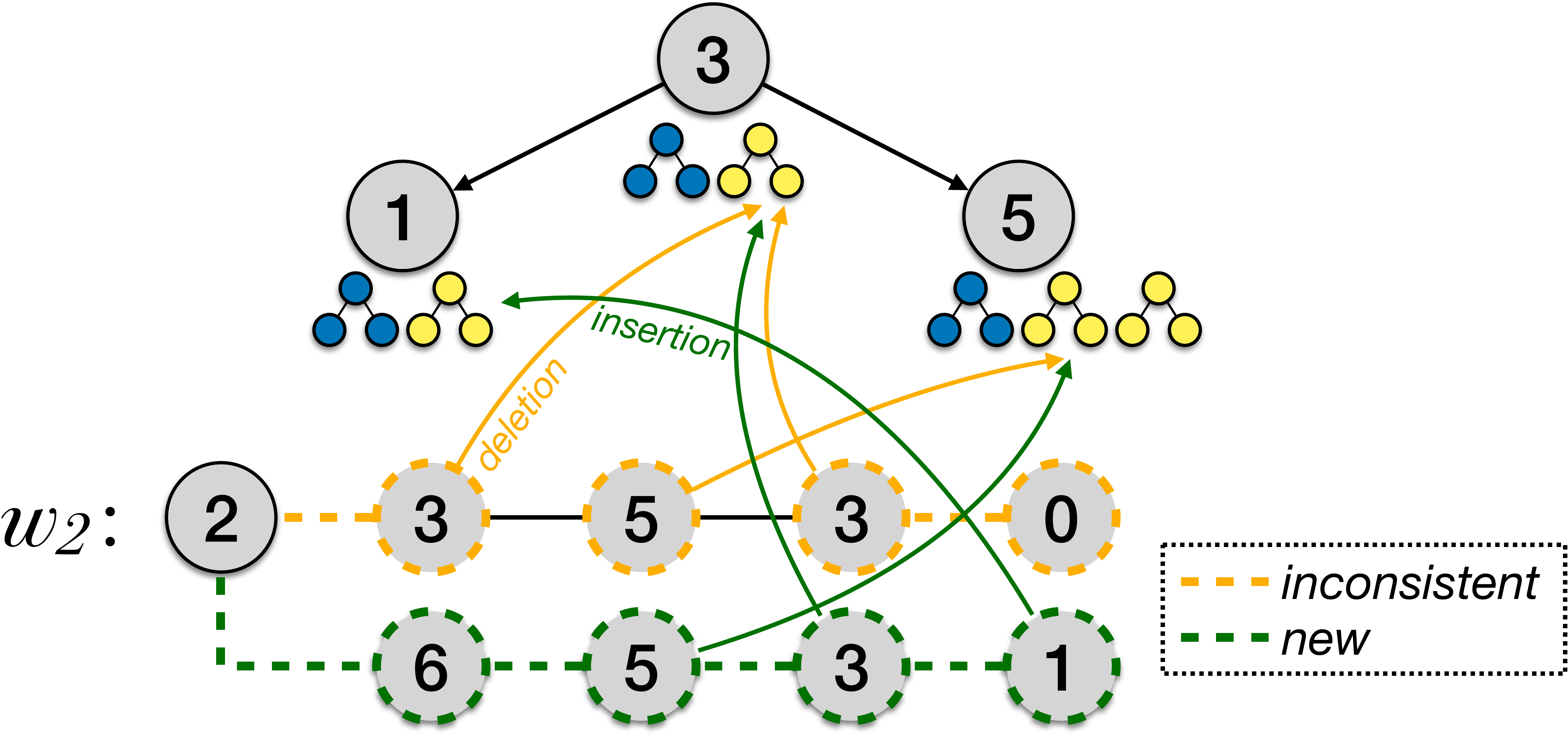}
	\caption{Running example of Batch Walk Update algorithm.} 
	\label{fig:running-example-bwu}
\end{figure}

As a walk corpus must remain statistically indistinguishable, 
we adopt the following update policy:\footnote{Yet, note that the walk update policy is orthogonal to our algorithm.} 
We update both inconsistent and invalid walks by deleting and re-sampling all affected vertices of an affected walk starting from the first affected vertex at position $p_{min}$ until the last vertex at position $p_{l}$.
For instance, Figure~\ref{fig:running-example-bwu} shows that walk $w_2$ of our running example is inconsistent, and we thus need to re-sample from its second vertex onward.

\wharf allows a vertex in the hybrid-tree to have more than one walk-tree version, each containing walk triplets corresponding to a 
distinct batch edge update. 
\wharf stores the walk-tree versions under a vertex in the order that it creates them. 
Furthermore, it utilizes a \textsc{Merge} operation that consolidates all the walk-tree versions of a vertex into a single one, after evicting all the obsolete walk-triplets. 
In specific, it scans the hybrid-tree in parallel, consults the MAV to check which triplets are still valid,
and removes the obsolete ones. 
\wharf uses an on-demand policy, which corresponds to merging walk-trees only when requested, e.g.,~when a downstream operation requests for the random walks.
Different policies with different throughput-memory trade-offs are also possible but we choose the on-demand one because it achieves the highest throughput. 
\small
\begin{algorithm} [t]
	\caption{\textsc{BatchWalkUpdate}}
	\label{alg:batch-walk-update}
	
	\begin{algorithmic}[1]
	    \State \textbf{Input:} hybrid-tree $\mathcal{H}$,  map $MAV$, walk model $\mathcal{M}$
	    \State \textbf{Output:} updated hybrid-tree $\mathcal{H}$

	\State {\textbf{do in parallel}}
	    \State{\hspace{\algorithmicindent}// Sampling of new walk parts}
        \State{\hspace{\algorithmicindent}\textbf{for}\textbf{ $w \in MAV$} \textbf{do in parallel}}
		\Comment { for each affected walk}
		\State\hspace{\algorithmicindent}\hspace{\algorithmicindent} $new\_ver = v_{w, p_{min}}$
	        \State{\hspace{\algorithmicindent}\hspace{\algorithmicindent}\textbf{for}\textbf{ $p=p_{min},\ldots,l-1$}\textbf{ do}}
	        \Comment {rewalk from $p_{min} < l$}
		     
	            \State\hspace{\algorithmicindent}\hspace{\algorithmicindent}\hspace{\algorithmicindent}$new\_ver = \textsc{sampleNext}(new\_ver, \mathcal{H}, \mathcal{M})$
		   
		    \State {\hspace{\algorithmicindent}\hspace{\algorithmicindent}\hspace{\algorithmicindent}$\mathcal{I} = \mathcal{I} \cup \textsc{EncodeTriplet}(w, p, new\_ver)$}
		    
	       \State{\hspace{\algorithmicindent}\hspace{\algorithmicindent}\textbf{ end for}}
		    \State {\hspace{\algorithmicindent}\hspace{\algorithmicindent}$\mathcal{I} = \mathcal{I} \cup \textsc{EncodeTriplet}(w, l, new\_ver)$}
  	        \Comment {last vertex}
		\State{\hspace{\algorithmicindent}\textbf{ end for in parallel}}
		
	    \State{\hspace{\algorithmicindent}// Evict obsolete walk triplets}
		  \State{\hspace{\algorithmicindent}\textbf{if $demanded$ then}}
	    \State {\hspace{\algorithmicindent}\hspace{\algorithmicindent}\textsc{Merge}($\mathcal{H}$, $MAV$)}
    \State{\textbf{end do in parallel}}
    
	\State {\textsc{MultiInsert}($\mathcal{H}, \mathcal{I})$}
        
	\State\Return{$\mathcal{H}$}

	\end{algorithmic} 
\end{algorithm}
\normalsize

Algorithm~\ref{alg:batch-walk-update} shows the pseudocode of the process to update random walks.
It takes as input the hybrid-tree, $\mathcal{H}$,
the $MAV$,
and the walk model $\mathcal{M}$. 
For each affected walk that appears in the $MAV$ (Line~5), we first initialize the vertex pointer for re-walking $w$ (Line~6).
We, then, re-walk from this vertex pointer, i.e.,~the vertex at the minimum affected position (Line~7), 
and fill 
the insertion accumulator $\mathcal{I}$ (Lines~8-11). 
In detail, if we are not yet at the end of $w$ (Line~11), we sample a new vertex with the new transition probability (Line~9).  
Note that depending on the utilized walk model $\mathcal{M}$ we must initialize the MH samplers~\cite{icde21-uninet} accordingly. For instance, when we use DeepWalk only the current vertex is needed for sampling the next vertex, however, when we use node2vec we need to access the previous vertex id before $p_{min}$ for initializing the samplers. 
Subsequently, we encode the triplet of the new vertex (Lines~9 \&~11). 
As a result, 
$\mathcal{I}$ maintains all the encoded walk triplets (i.e.,~integer values)
that should be 
inserted grouped by vertex  identifier.  
In Figure~\ref{fig:running-example-bwu}, we see $w_2$, which is affected, and an excerpt of the hybrid-tree, namely, vertices $v_1,v_3,v_5$. 
The newly sampled vertices (circled in green) are converted into walk triplets and then are batch-inserted to the corresponding walk-trees of $v_1$, $v_3$, and $v_5$ (green arrows indicate insertion operations). 

While we are running the re-walking process, 
we run the \textsc{Merge} process in the background to
delete the obsolete walk parts from the walk-trees (Lines~14~\&~15).
In the example of Figure~\ref{fig:running-example-bwu}, the merge process is triggered in parallel with the sampling of new vertices to delete $w_2$'s inconsistent walk triplets (circled in orange) from the corresponding walk-trees, e.g.,~of $v_3$ and $v_5$ (orange arrows show deletion operations). 
Note that merge consolidates potentially multiple walk-tree versions, e.g.,~of $v_5$ (more than one yellow trees). 
Finally, we apply the batch 
insertions of the newly generated walk parts 
(Line~17). 
Note that we  
use
the \textsc{MultiInsert}
method for applying batch updates to $C$-trees~\cite{pldi19-aspen}. 
As an outcome, the algorithm produces the hybrid-tree $\mathcal{H}$ 
with the updated walk corpus that is statistically indistinguishable from a 
corpus 
generated from scratch.

\subsection{Complexity and Correctness}
%
Let us now elaborate on the time complexity of Algorithm~\ref{alg:batch-walk-update} in terms of number of walk triplets that are inserted and deleted. 
We have to update $a = |MAV|$ affected walks. 
Precisely, we should insert in the hybrid-tree after re-walking, $|\mathcal{I}| = \sum_{i=1}^{i=a}(l-p_{min}^{i})=O(a \times l)$ walk-triplets, where $p_{min}^i$ is the minimum affected position of the $i^{th}$ affected walk in the $MAV$. 
The batch insertion is done by the \textsc{MultiInsert}~\cite{pldi19-aspen}, which has a complexity of $O(|\mathcal{I}|\log|W|)$ work overall, and $O(\log^3 |W|)$ depth, where $|W| = n * n_w * l$ is the total number of walk-triplets in a walk corpus, $n$ is the number of vertices in the graph, $n_w$ is the walks per vertex, and $l$ the length of each walk.
Therefore, the complexity 
of Algorithm~\ref{alg:batch-walk-update} is $O(|\mathcal{I}|\log|W|)$ work and $O(\log^3 |W|)$ depth. 
An inverted index-based solution 
needs 
$\Theta(\sum_{i=1}^{i=a}p_{min}^{i})$ time to construct the $MAV$, as it traverses an affected walk from its first vertex till its $p_{min}$. 
Additionally, it has to update the affected walk parts like \wharf, and thus, the total complexity is $\Theta(a \times l) = \Omega(|\mathcal{I}|)$. Hence, the complexity of an inverted index solution is 
greater than that of \wharf. 

\begin{theorem}[Correctness]
\wharf 
updates the random walks in a walk corpus, 
such that they remain statistically indistinguishable.  
\end{theorem}
\begin{proof}[Proof sketch]
\label{prop:proof-statistically-indistinguishability}
Fix a random walk $w \in \mathcal{W}$, where $\mathcal{W}$ is the maintained walk corpus. 
Let $(s,d)$ be an undirected edge that gets inserted (w.l.g.) into the graph. 
We discern the following two cases:
\begin{packed_enum}
\item $s \notin w$ and $d \notin w$. 
As none of the two endpoints of the incoming edge are ``covered'' by $w$, 
the transition probabilities with which $w$ was sampled, using an walk model of up to second-order,
do not change.
Specifically, in first-order walks (e.g.,~DeepWalk), a vertex $v$ uniformly samples one of its neighbors as the next vertex in $w$.
Thus, the transition probabilities between vertices in $w$ stay intact and hence $w$ remains valid. 
In second-order walks (e.g.,~node2vec), a vertex $v$ (with $v_{prev}$ as the previous vertex in $w$) non-uniformly samples one of its neighbors as the next vertex in $w$: It does so with a probability that depends on whether the next vertex is (or isn't) connected with $v_{prev}$, or the next vertex is actually $v_{prev}$~\cite{kdd16-node2vec}. As $s \notin w$ and $d \notin w$, the transition probabilities of $w$ remain intact.

\item $s \notin w$ but $d \in w$.
When one endpoint is not ``covered'' by $w$, \wharf incorporates $d$ into the edge-tree of vertex $s$ as well as $s$ into the edge-tree of vertex $d$ right after the insertion of the undirected edge $(s,d)$.
\wharf also checks the walk-tree of $s$, and the one of $d$, where it finds the corresponding walk triplet belonging to $d$, and thus, identifies that $w$ is affected and proceeds to update it. 
This holds for both first- and second-order walks.

Notice that, a similar reasoning holds in case ($s$, $d$) is a deletion. 
\end{packed_enum}
\end{proof}

\section{Experimental Evaluation}
\label{sec:experiments}
We evaluate \wharf using a variety of large-scale real-world and synthetic graphs and 
investigate: how efficient it is in terms of throughput, latency, and space; how it scales to large graphs and batch sizes; how it behaves in the presence of data skew; how its range search and merge policy drives its performance; and whether it can enable high accuracy on downstream 
tasks. 

\begin{myboxi}
Overall, we found that \wharf: (i)~has superior throughput and latency compared to the baselines making it suitable for streaming scenarios, (ii)~is space-efficient thanks to its walk-tree structure, (iii)~scales with both the graph and batch size and is robust to skew, and~(iv) is always better than recomputing the walks from scratch in contrast to the baselines.
\end{myboxi}

\subsection{Setup} 
\label{subsec:setup}

\textbf{Hardware.} 
We ran our experiments on a server with a $24$-core Intel(R) Xeon(R) Gold $6126$ CPU @ $2.60$GHz with 2-way hyperthreading and $1.5$TB of main memory. 
Our prototype uses the work-stealing scheduler in \cite{pldi19-aspen}, which is implemented similarly to Cilk for parallelism. 
We compiled our programs with the \texttt{g++} compiler (version 9.2.1) having the \texttt{-O$3$} flag 
and ran all our experiments five times and report the average.

\noindent \textbf{Implementation.} 
We implemented \wharf in C++20 on top of Aspen~\cite{pldi19-aspen} and
used weight-balanced trees as the underlying balanced tree implementation \cite{spaa16-just-joins,ppopp18-pam}.
Note that Aspen's current implementation supports storing up to $64$-bit integers in $C$-trees. Consequently, each 
Szudzik operand in \wharf should be up to $32$ bits. 
We stress that this is not a limitation of \wharf, but of the $64$-bit implementation of Aspen (on which we built). 
Yet, \wharf can still support much larger graphs for PPR use cases where walks are shorter (around $5$-$15$ vertices long).
As explained in~\cite{vldb10-goel}, the theoretical guarantees are preserved for $n_w = 10$ and $l = 10$, so \wharf can scale to graphs with up to $(2^{32}-1)/100 \approx 42.94$M vertices, such as the \textit{Twitter} dataset~\cite{www10-twitter} as we show in our experiments. 


\small
\begin{table}[h!]
    \begin{center}        
    \caption{Datasets Statistics.}
     \vspace{-0.4cm}
        \begin{tabular}{l|r|r|S}

            \hline
            
            \textbf{Graph} & \textbf{Num. Vertices} & \textbf{Num. Edges} & \textbf{Avg. Degree} \\

            \hline
            \hline
            
            \textit{com-YouTube} & 1,134,890 & 2,987,624 & 5.3       \\
            \textit{soc-LiveJournal} & 4,847,571 & 85,702,474 & 17.8 \\
            \textit{com-Orkut} & 3,072,627 & 234,370,166 & 76.2      \\
            \textit{Twitter} & 41,652,230 & 1,468,365,182 & 57.7     \\
            
            \hline
            
            
            
            
            
        \end{tabular}
        \label{tab:datasets}
  \end{center}
\end{table}
\normalsize

\noindent \textbf{Datasets.} 
We used four real-world (Table~\ref{tab:datasets}) and eight synthetic graph datasets:
\textit{Real Graphs} -- \textit{\textbf{com-Youtube}} is an undirected graph of the Youtube social network \cite{icdm12-jure}, \textit{\textbf{soc-LiveJournal}} is a directed graph of LiveJournal social network \cite{kdd06-kleinberg}, \textit{\textbf{com-Orkut}} is an undirected graph of Orkut social network \cite{icdm12-jure}, and \textbf{\textit{Twitter}} is a directed graph of the Twitter network, in which edges model the follower-followee relationship~\cite{www10-twitter};
\textit{Synthetic Graphs} -- We generated large-scale synthetic graphs sampled from the R-MAT model \cite{sdm04-rmat}. 
Specifically, we used the TrillionG\footnote{https://github.com/chan150/TrillionG} \cite{sigmod17-trillionG} tool to generate  
\textit{\textbf{Erd\H{o}s R\'{e}nyi, er-$k$,}} graphs  with $2^k$ nodes, uniformly distributed edges and witha an average vertex degree of $100$ by setting the R-MAT parameters to $a = b = c = d = 0.25$.
Additionally, we varied $k$ from $16$ to $22$ to evaluate the scalability of \wharf.
We also generated a set of skewed graphs, \textit{\textbf{sg-$s$}}, with $2^{20}$ nodes with an average degree of $10$, 
while varying the skew.
We set the R-MAT parameters ($a$, $b$, $c$, $d$) so that the number of edges in the bottom-right part of the matrix is about $s$ times 
the top-left part of the matrix.
We set $b = c = 0.25$. 
Thus, when $s = 1$, there is no skew, while if $s > 1$,~R-MAT generates power-law graphs. 
We varied $s$ from $1$ to $7$ with a step of $2$. 
\noindent \textbf{Baselines.} 
As there is no system that maintains streaming random walks,
we compared \wharf with approaches proposed in~\cite{arxiv19-incremental-node2vec, cnta19-evonrl},
which use an inverted index for maintaining walks. We call this baseline Inverted Index-based (\texttt{II-based}).
Specifically, II-based maintains the walks separately from the graph in sequences of vertices stored in vectors. 
We used a dictionary for storing the walks, where the key is the walk id and the value is the walk sequence. 
Additionally, II-based maintains an inverted index that relates a vertex id with the set of walk ids it participates. 
For fairness reasons, we implemented a fully parallel version of II-based by utilizing concurrent hashtables\footnote{\url{https://github.com/efficient/libcuckoo}}. 
We also used a \texttt{Tree-based} baseline, which stores the walk-triplets into parallel balanced binary trees (a.k.a.~parallel augmented maps)~\cite{ppopp18-pam}; a structure that provides highly parallel operations and upon which $C$-trees are built. 
\subsection{Overall Performance}
\label{subsec:overall}

\noindent{\bf Throughput \& Latency.} We compare \wharf with II-based and Tree-based in terms of throughput, i.e.,~number of updated walks per second, and latency, i.e.,~the average time for updating one walk, when updating random walks.
If not stated otherwise, we use the DeepWalk~\cite{kdd14-deepwalk} walking model with the default parameters, i.e.,~$n_w = 10$ walks per vertex of length $l=80$. 
For this experiment, we used the real graphs and generated 
walks of default length 
for the first three real datasets, whereas $l=10$ for twitter dataset. 
We also produced batches of $10,000$ edges, which we sampled based on the R-MAT~\cite{sdm04-rmat} model with parameters $a = 0.5$, $b = c = 0.1$, and $d = 0.3$ to induce graph updates as in~\cite{pldi19-aspen}. 
We inserted $10$ such batches 
in total and report the average
throughput and latency.

Figure~\ref{fig:real-throughput-latency} illustrates 
the throughput and latency 
results.
We observe 
that \wharf is superior to both baselines in all cases. It achieves up to $\sim\!\! 2.6\times$ higher throughput 
and $2\times$ lower latency, which is crucial for streaming applications. 
This is thanks to its parallel \textsc{MultiInsert} and \textsc{Merge} operations that
update different parts of the walk corpus on the hybrid-tree simultaneously. 
The advantage of \wharf is more evident for the datasets where we used larger walk lengths, such as \textit{Livejournal}.
Contrary to what one may think the Szudzik encoding/decoding function calls required only $3.66$\% of the total time for walk updates in \textit{com-YouTube}, $10.055$\% in \textit{soc-LiveJournal}, $7.797$\% in \textit{com-Orkut}, and  $12.815\%$ in Twitter.

Also, \wharf's default (on-demand) policy for merging allows it to achieve maximum throughput by only merging at the last batch.\footnote{As expected, there exists a throughput-memory trade-off: one can achieve higher throughput at the price of a higher memory footprint by merging less frequently.} 
On the contrary, even though II-based uses parallelism, maintaining the walks in sequences that should be scanned to get updated, leads to reduced throughput 
due to thread contention. 
Tree-based achieves poor throughput because of re-walking obsolete parts of affected walks to remove them. 

During our experiments we observed that the total time and throughput of updating the walks 
for edge deletions is within $10\%$ of the time required for edge insertions. 
To illustrate this, we generated $5$ batches of edges and for each batch we alternately applied insertion and consequently deletion, where each triggers updates of afffected walks.

Figure~\ref{fig:mixed-workload} shows the throughput of updating walks on \textit{soc-LiveJournal} due to insertions (I) and deletions (D) for batches of $10$K and $100$K edges.
As shown, the throughput of deletions is similar to that of insertions. 
We got similar results for the other real datasets. 
In the sequel, we show results only for edge insertions.
 
We thus conclude that \emph{\wharf is superior than the baselines and its high throughput makes it suitable for streaming graphs.} 

\begin{figure}[t]
  \begin{subfigure}[t]{0.23\textwidth} 
    \includegraphics[width=\textwidth]{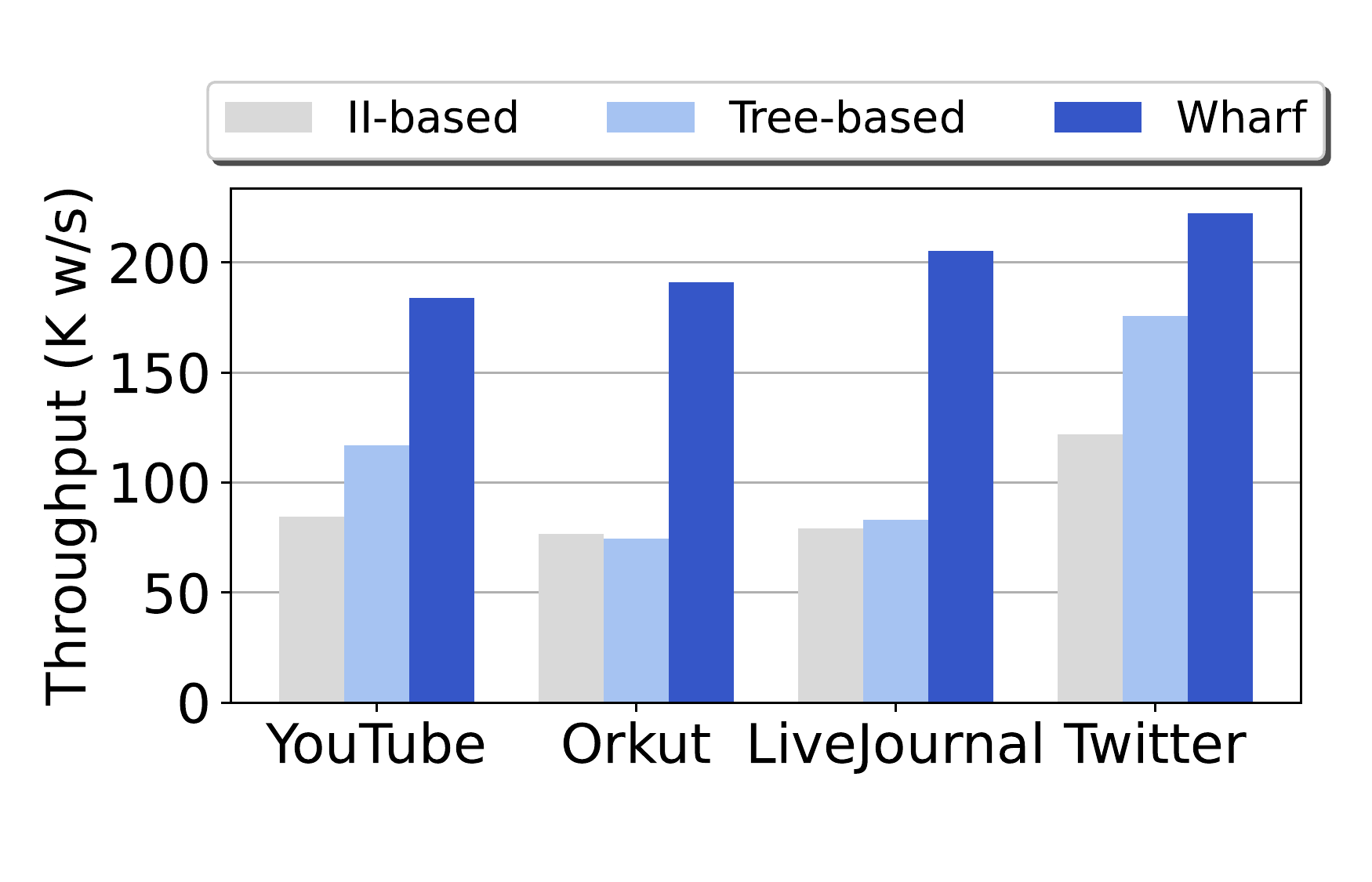}
    \caption{Throughput} 
    \label{fig:real-throughput}
  \end{subfigure}
    \begin{subfigure}[t]{0.23\textwidth} 
    \includegraphics[width=\textwidth]{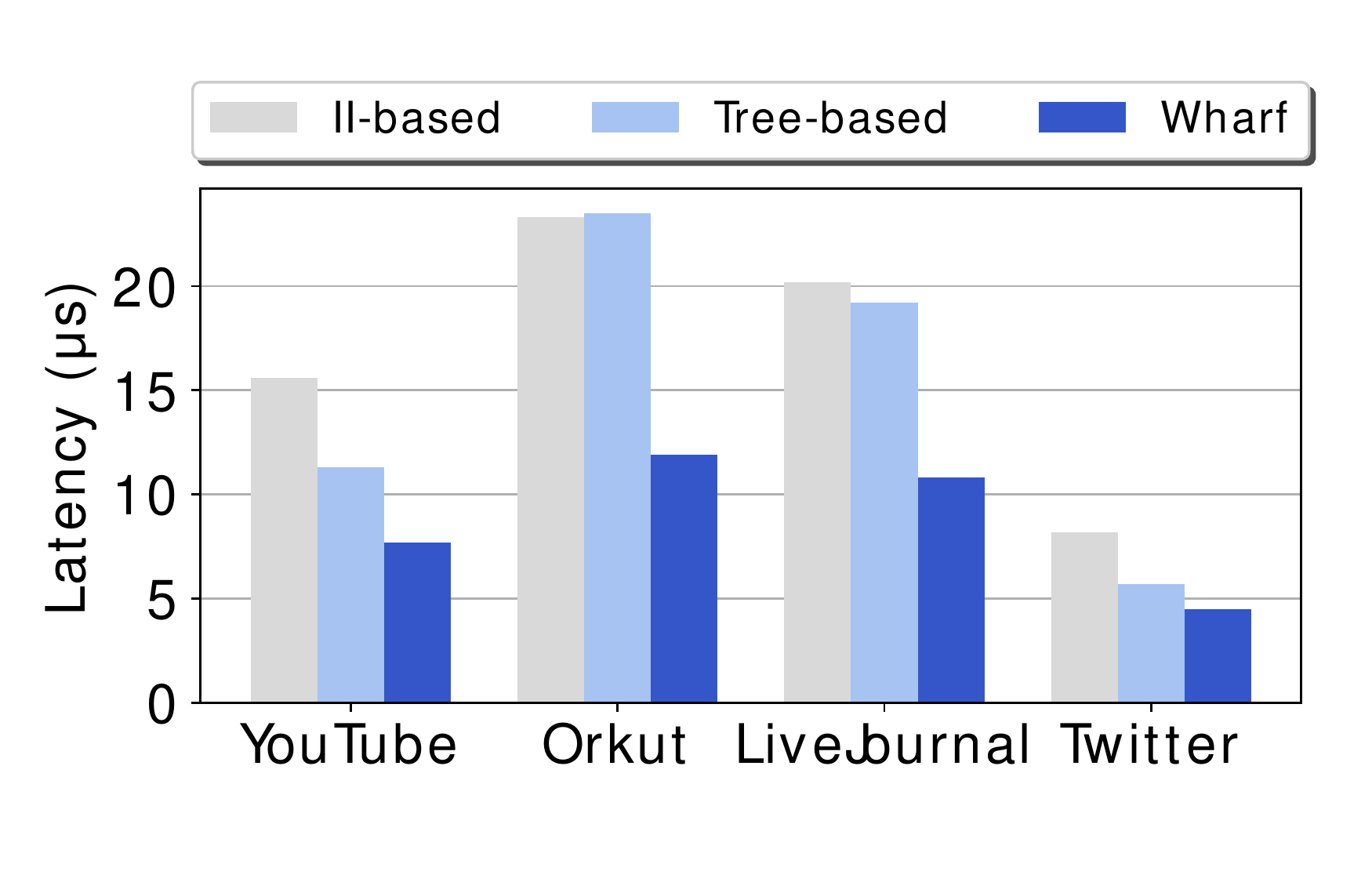}
    \caption{Latency} 
    \label{fig:real-latency}
  \end{subfigure}
  \vspace{-0.2cm}
  \caption{Performance of \wharf on real graphs.}
  \label{fig:real-throughput-latency}
  \vspace{-0.4cm}
\end{figure}

\begin{figure*}[t!]
\begin{minipage}{0.25\textwidth}
    \centering 
    \begin{subfigure}{\linewidth}
    \vspace{-0.4cm}
    \centering
    \includegraphics[width=\linewidth]{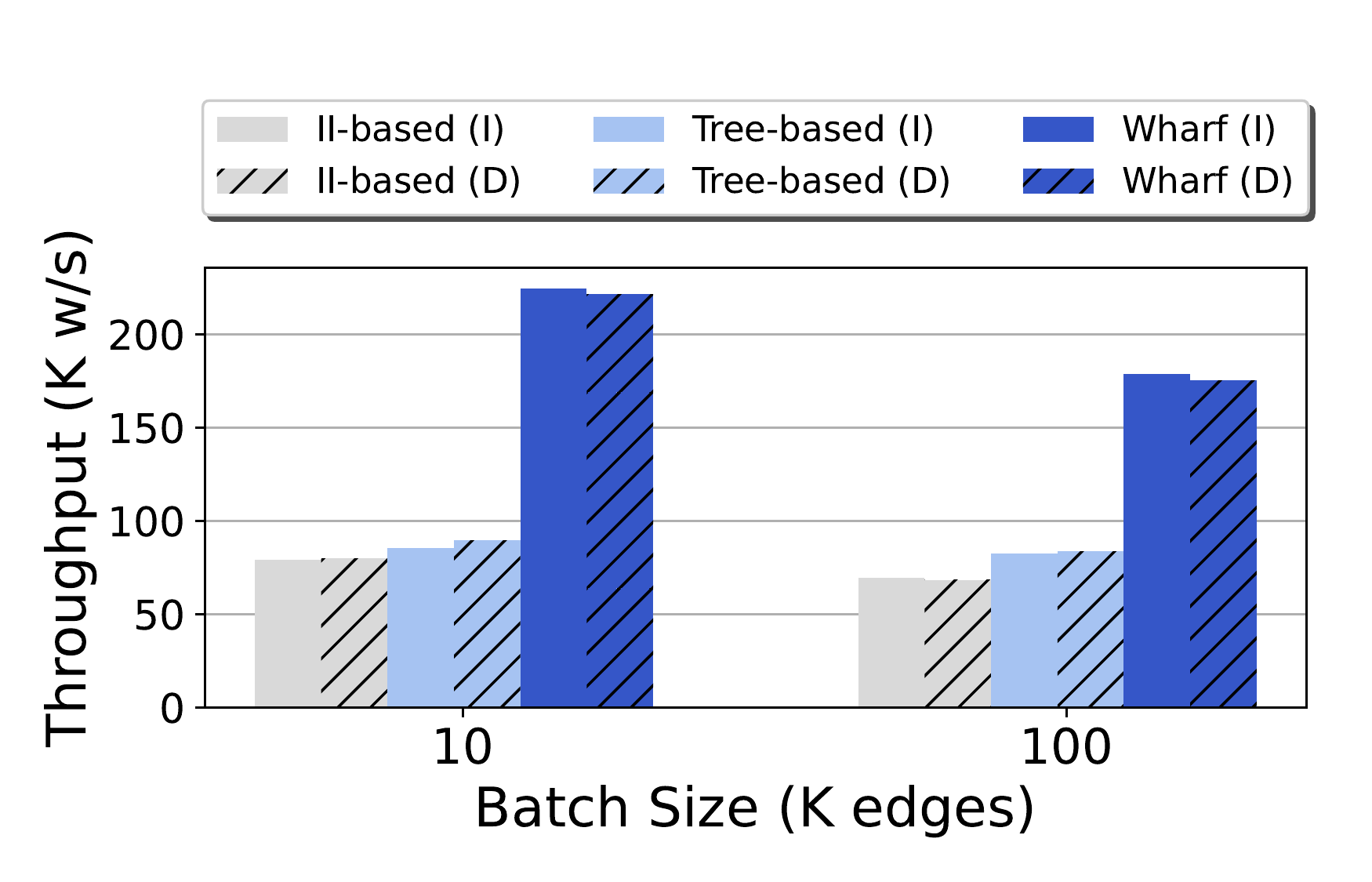} 
    \caption{\textit{LiveJournal}}
    \vspace{-0.2cm}
    \end{subfigure}
    \caption{Mixed workload.}
    \label{fig:mixed-workload} 
\end{minipage}
\begin{minipage}{0.74\textwidth} 
    \begin{subfigure}{0.32\textwidth}
    \centering
    \includegraphics[width=\linewidth]{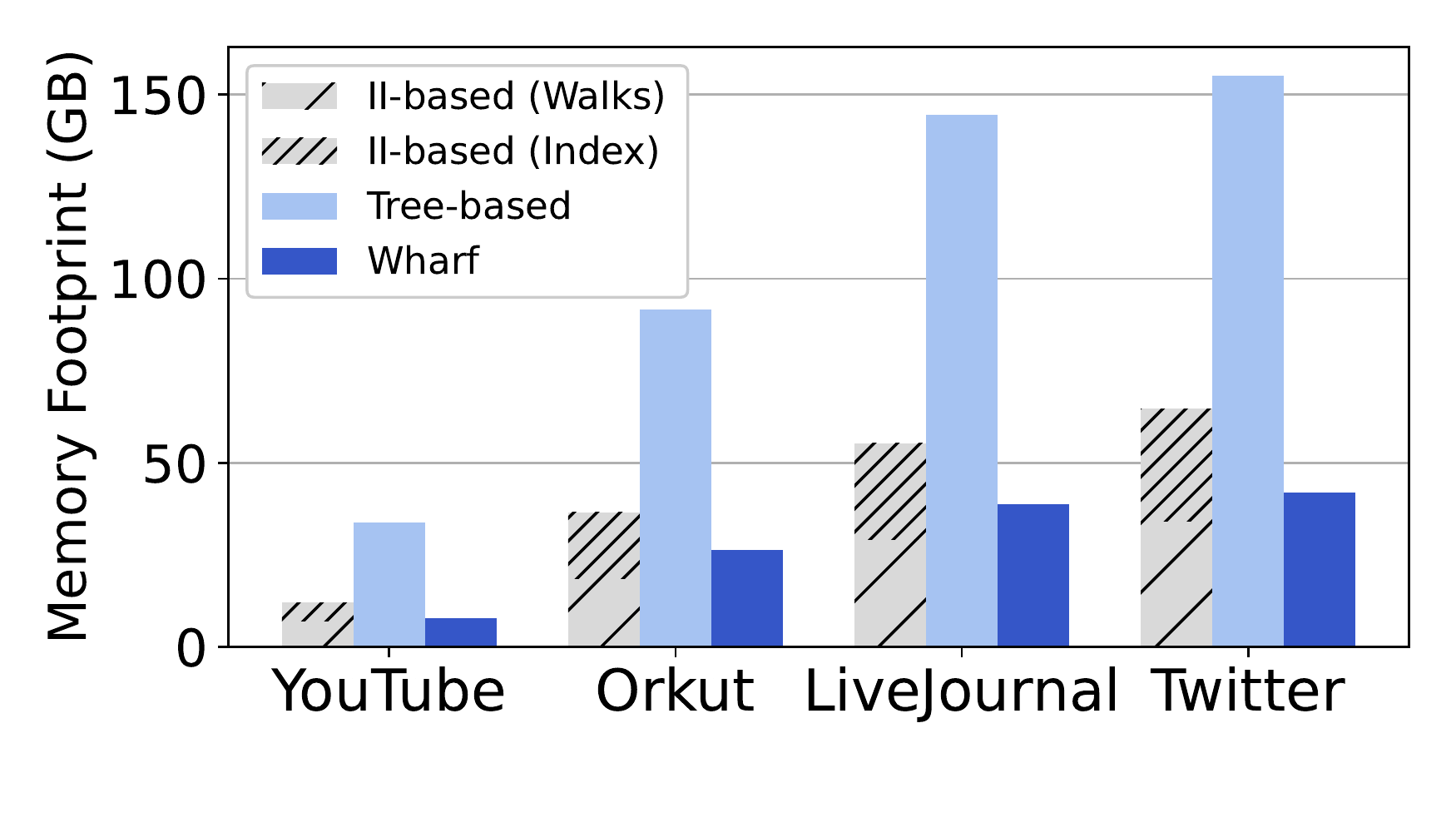}  
    \captionsetup{justification=centering}
    \caption{Real Datasets}
    \label{fig:total-memory-footprint}
    \end{subfigure}
    \begin{subfigure}{0.32\textwidth}
    \centering
    \includegraphics[width=\linewidth]{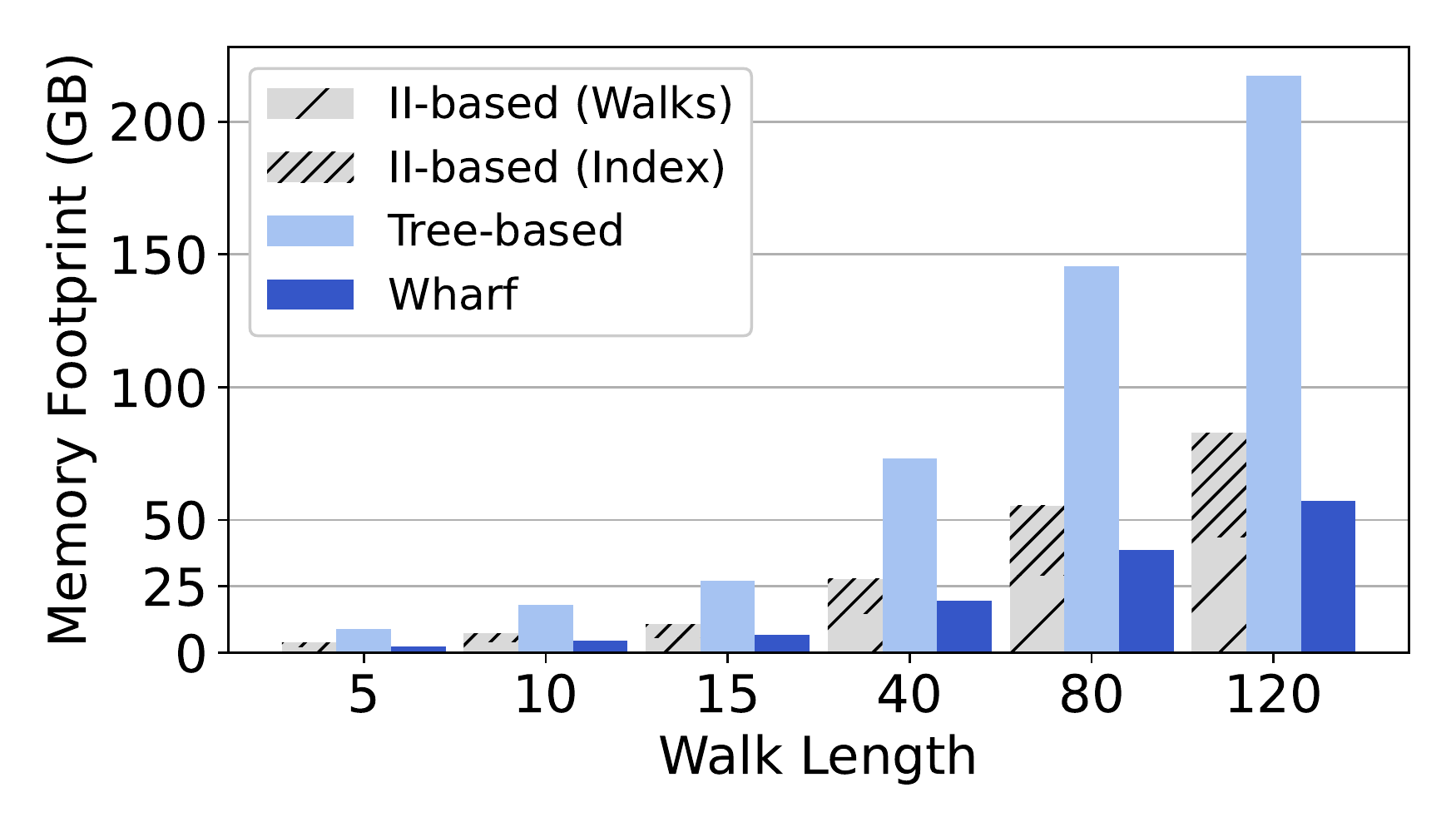}
    \captionsetup{justification=centering} 
    \caption{\textit{LiveJournal}, varying $l$, $n_w=10$}
    \label{fig:varying-l}
    \end{subfigure}
    \begin{subfigure}{0.32\textwidth}
    \centering
    \includegraphics[width=\linewidth]{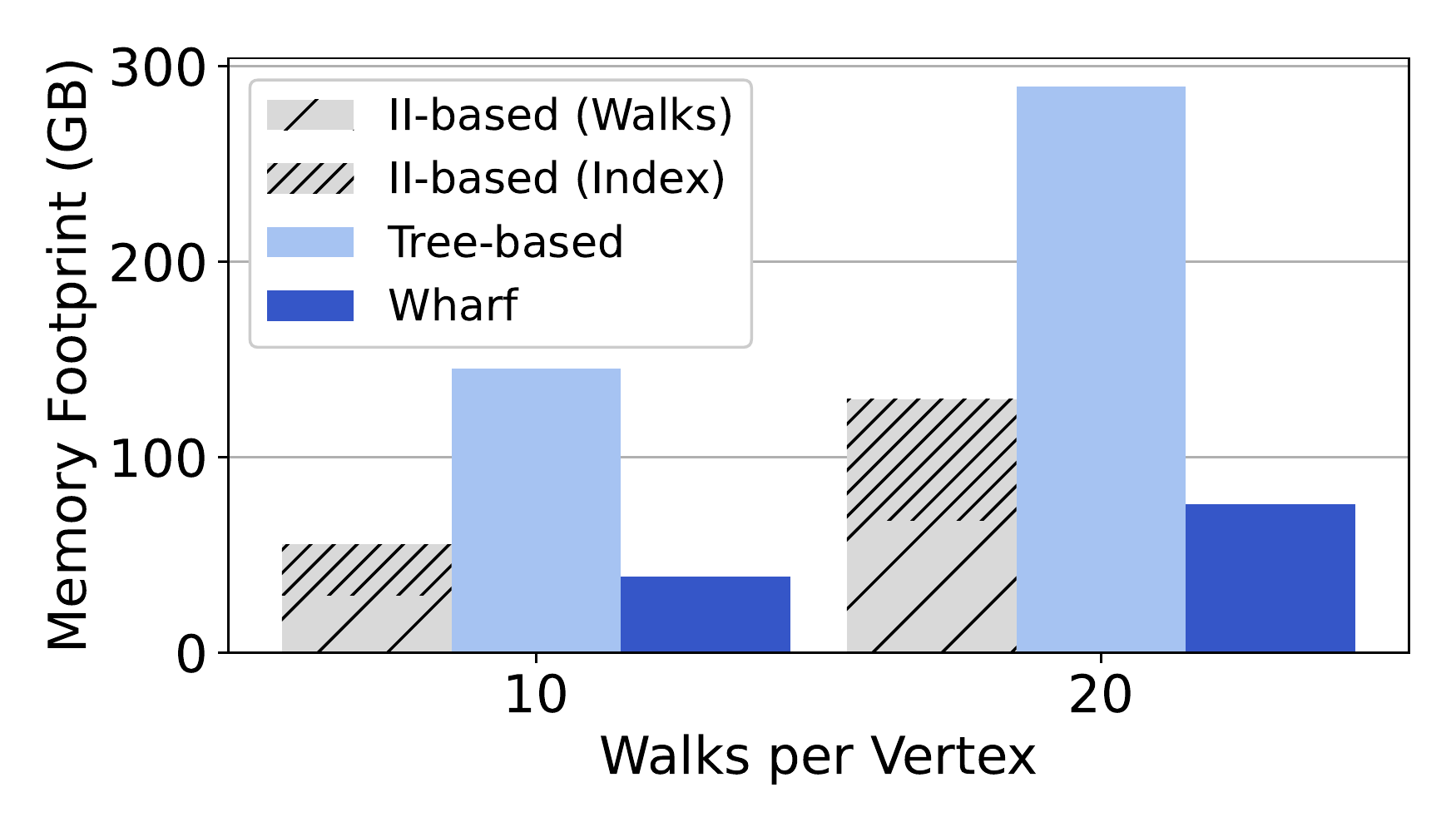}
    \captionsetup{justification=centering}
    \caption{\textit{LiveJournal}, varying $n_w$, $l=80$}
    \label{fig:varying-w}
    \end{subfigure}
\caption{Memory footprint of \wharf on real graphs.}
\label{fig:SA}
\end{minipage}\hfill
\vspace{-0.4cm}
\end{figure*}

\vspace{0.1cm}
\noindent{\bf Memory Footprint.} We also 
compare the memory footprint of \wharf with that of II-based and Tree-based.
We aim to compare the compression capabilities of \wharf's data structure irrespective of the merge policy, and thus, we report the memory needed after merging.
Note that, we also compare \wharf with KnightKing~\cite{sosp19-knightking}: 
\wharf requires less than 30\% more storage.
Yet, we
focus only on II-based and Tree-based, because KnightKing (i)~requires to build the entire graph after every single update,
(ii)~does not store random walks in any structure but outputs them in raw files, and (iii)~offers neither any efficient search nor update capabilities for the walks.

Figure~\ref{fig:total-memory-footprint} shows the total space that \wharf needs to store the walk corpus.
We show a breakdown of the memory needed by II-based to store the walks and the memory needed to store the inverted index.
We observe that 
\wharf can store the walks with up to $1.7\times$ less space than II-based. 
Especially, we observe that \wharf stores its walks 
using only $10.22-29.54\%$ more space than the space II-based uses for storing {\em only} the walks.
For instance, in \textit{soc-Livejournal}, 
II-based requires $29.25$ GB for the walks and $26.09$ GB for the inverted index, whereas \wharf stores the walks, which are implicitly indexed, using $38.83$ GB. 
As for the Tree-based, we observe that its memory footprint is $\sim\!\!3.5-4.4\times$ higher than \wharf's, because it stores the walk-triplets without any compression.

Figures~\ref{fig:varying-l} and~\ref{fig:varying-w} 
illustrate the total memory footprint when varying the walk length $l$ and the number of walks per vertex $n_w$. We report the results only for the \textit{soc-LiveJournal} dataset because we observed the same behaviour for the other real graphs.
Figure~\ref{fig:varying-l} shows a linear behaviour in terms of space consumption with respect to the walk length.
\wharf requires on average $\sim\!\!1.5\times$ less space than II-based
and $\sim\!\!3.76\times$ less space than Tree-based.
Figure~\ref{fig:varying-w}
shows again a linear behaviour with respect to $n_w$, i.e.,~ \wharf has on average $\sim\!\!1.6\times$ smaller  memory footprint compared to II-based and $\sim\!\!3.77\times$ smaller than Tree-based.
These space-savings are thanks to the use of pairing functions in combination with differential encoding in the chunks of walk-trees. 
The reader might think of applying a simple compression technique in II-based, but existing techniques for compressing inverted indexes are neither trivial nor suitable for dynamic data~\cite{pibiri2020techniques}.
Therefore, we conclude that \emph{our proposed walk-tree structure enables \wharf to store an indexed walk corpus space-efficiently.}


Based on all the results above, we decided to discard the Tree-based baseline, and keep only the II-based baseline, in the subsequent experiments: II-based is comparable to Tree-based, in terms of throughput, while occupying much less space.

\subsection{Scalability}
\label{subsec:scalability}


\noindent{\bf Batch Size.} We now demonstrate \wharf's scalability when varying the batch size for edge insertions. We produced batches with sizes $10$, $25$, $50$, $75$, and $100$ thousand edges that we inserted in \textit{com-Orkut}. Notice that we omit the results for the other two real datasets because they follow the same trend as for \textit{com-Orkut}. 

\begin{figure}[t]
  \begin{subfigure}[t]{0.23\textwidth} 
    \includegraphics[width=\textwidth]{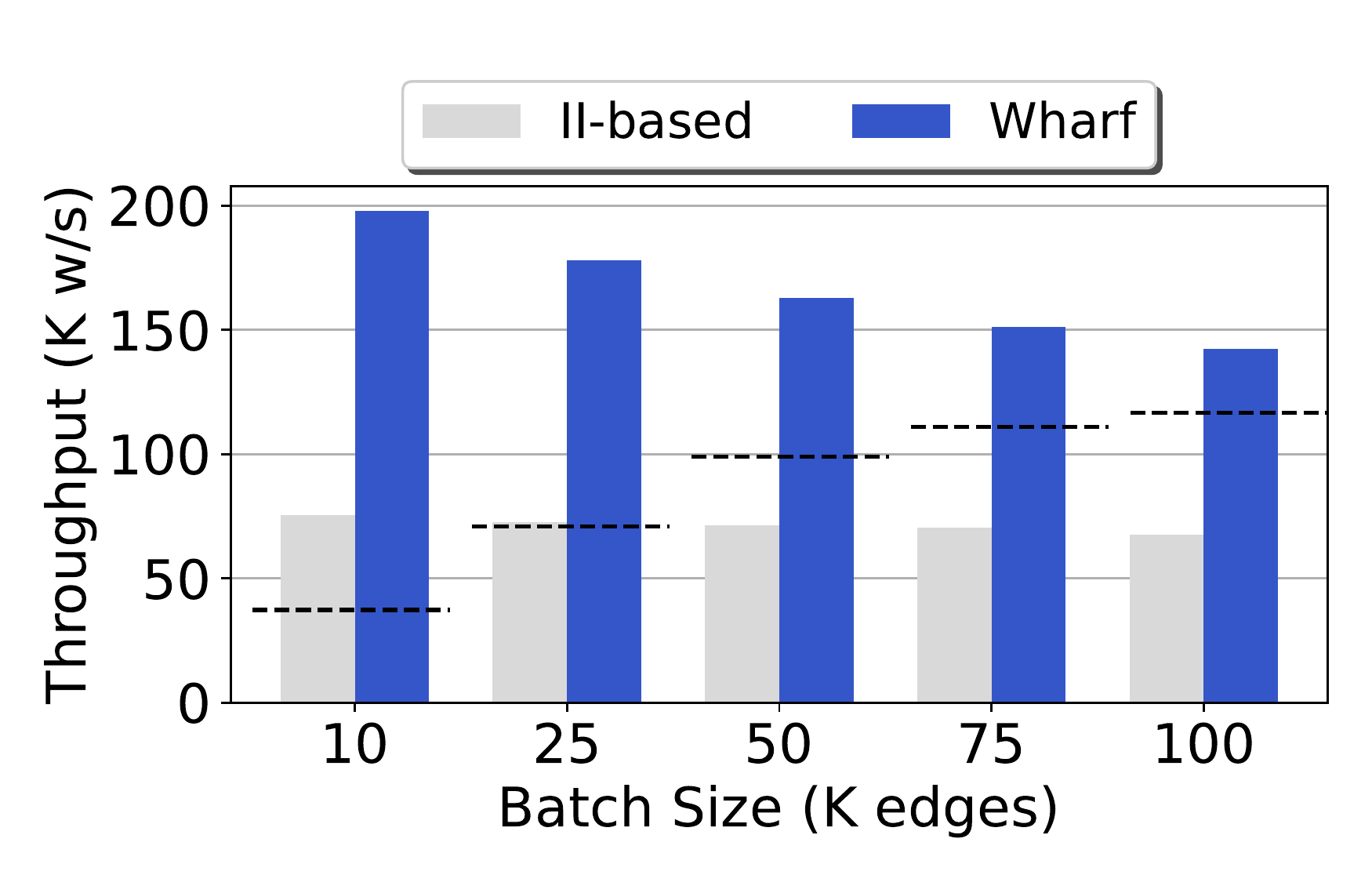}
    \caption{Throughput}
    \label{fig:orkut-throughput-varyingB}
  \end{subfigure}
  \begin{subfigure}[t]{0.23\textwidth} 
    \includegraphics[width=\textwidth]{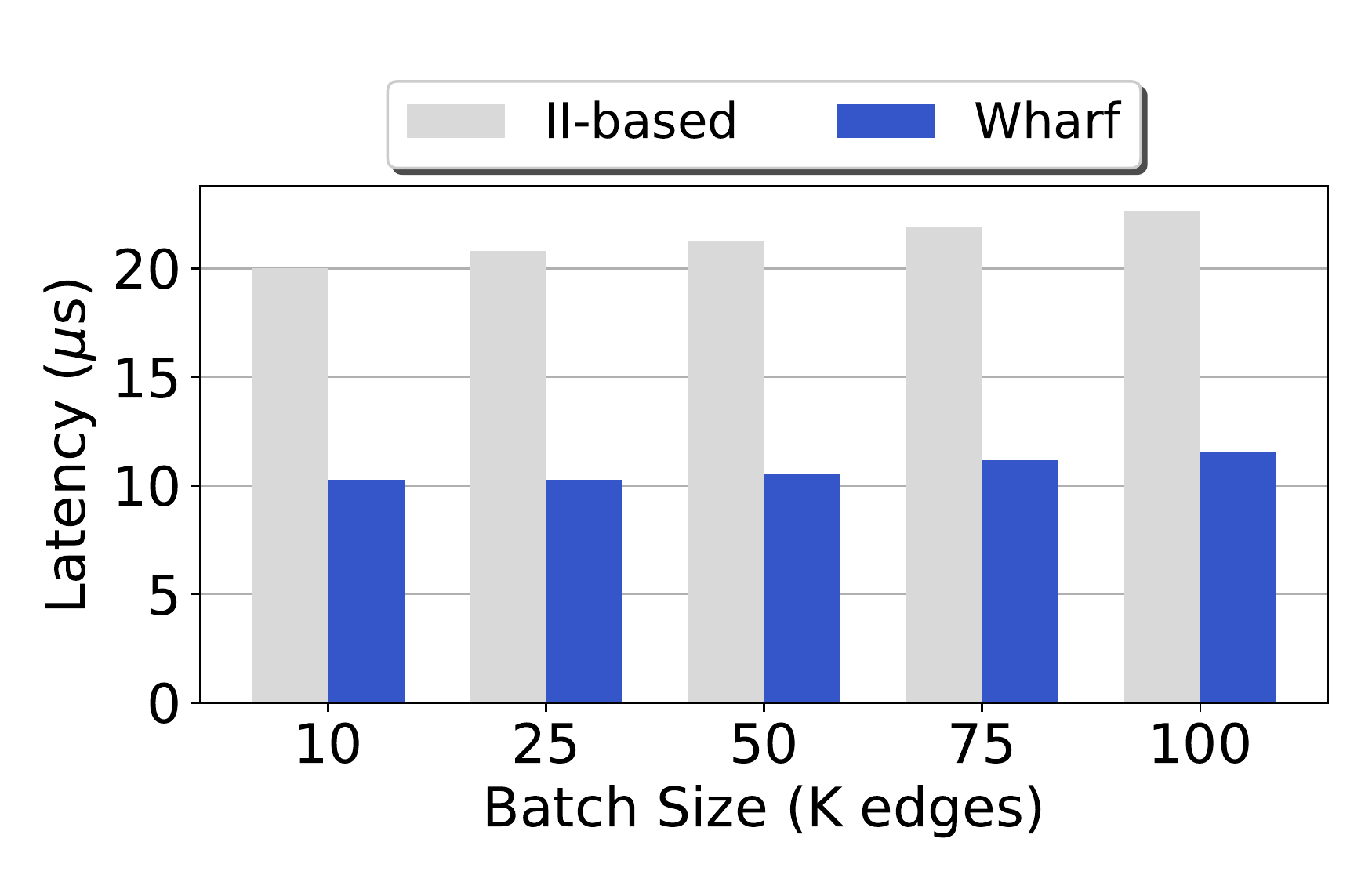}
    \caption{Latency}
    \label{fig:orkut-latency-varyingB}
  \end{subfigure} 
  \vspace{-0.2cm}
  \caption{Scalability as the batch size increases on \textit{com-Orkut}.}
  \label{fig:orkut-throughput-latency-varyingB}
  \vspace{-0.2cm}
\end{figure}

Figure~\ref{fig:orkut-throughput-varyingB} illustrates the throughput results, where the black horizontal lines represent the minimum throughput required to generate the walks from scratch.
We observe that \wharf is always better than recomputing the random walks from scratch, which is not the case for II-based for batch sizes larger than $25K$. 
This is because II-based
can only perform $72.3$K walk updates per second. 
In general, \wharf achieves up to $\sim\!\!2.6\times$ higher throughput than II-based.  We also observe that the throughput of both \wharf and of II-based decreases as the batch size increases, namely, $\sim\!\!22.4\%$ for \wharf and $\sim\!\!10.6\%$ for II-based going from batch size of $10$K to $100$K edges.
The reason is that the larger the batch size is the higher the average number of affected walks is.
\wharf's throughput decreases a bit more because 
(i)~the time to compute the $MAV$ increases as the walk trees get larger with larger batch sizes, and (ii)~the minimum position of the affected walks decreases.
During our experiments, we observed that as the batch size increases not only the number of affected walks increases, but also more and more walks are affected at an earlier point of the walk sequence. This leads to more work for updating the random walks. 
For instance, while inserting $50K$ edges leads to $\sim\!\!463$K affected walks from the first position, inserting $100$K edges leads to $\sim\!\!874$K affected walks from the first position. 

Figure~\ref{fig:orkut-latency-varyingB} illustrates the latency results. 
We observe that the latency of \wharf is $\sim\!\!2\times$ lower than the one of II-based. 
We also see that both \wharf's and II-based latency increases as the batch size increases 
because of the increased number of walks.
Yet, \wharf's latency stays considerably low thanks to its on-demand policy for merging.
We thus conclude that \emph{\wharf is more scalable than the baseline as it can keep pace even in highly streaming scenarios with many updates per batch}.



\vspace{0.1cm}
\noindent{\bf Input Graph Size.} We also study \wharf's scalability when varying the input graph size w.r.t. the number of vertices. For this experiment, we used the \textit{er}-graphs and fixed the batch size to $10K$ edges. 
Figures~\ref{fig:er-throughput} and~\ref{fig:er-latency} illustrate the throughput and latency, respectively. 
We see that \wharf achieves $1.9-2.5\times$ higher throughput than II-based, and $\sim\!\!1.8-2\times$ lower latency.
This is attributed to two things: (i)~the way \wharf stores the walks in the hybrid-tree that enables updating various parts of the corpus simultaneously, and (ii)~its on-demand merge policy.
Furthermore, we observe that, as the distribution of vertex degree in the $er$-graphs is uniform, the throughput of both solutions remains steady: $\sim\!\!115$K walks/second for \wharf and $\sim\!\!50$K walks/second for II-based.
Consequently, the number of affected walks increases proportionally to the graph size.
We conclude that \emph{it nicely scales with varying graph sizes}.

\begin{figure}[t]
  \begin{subfigure}[t]{0.22\textwidth}
    \includegraphics[width=\textwidth]{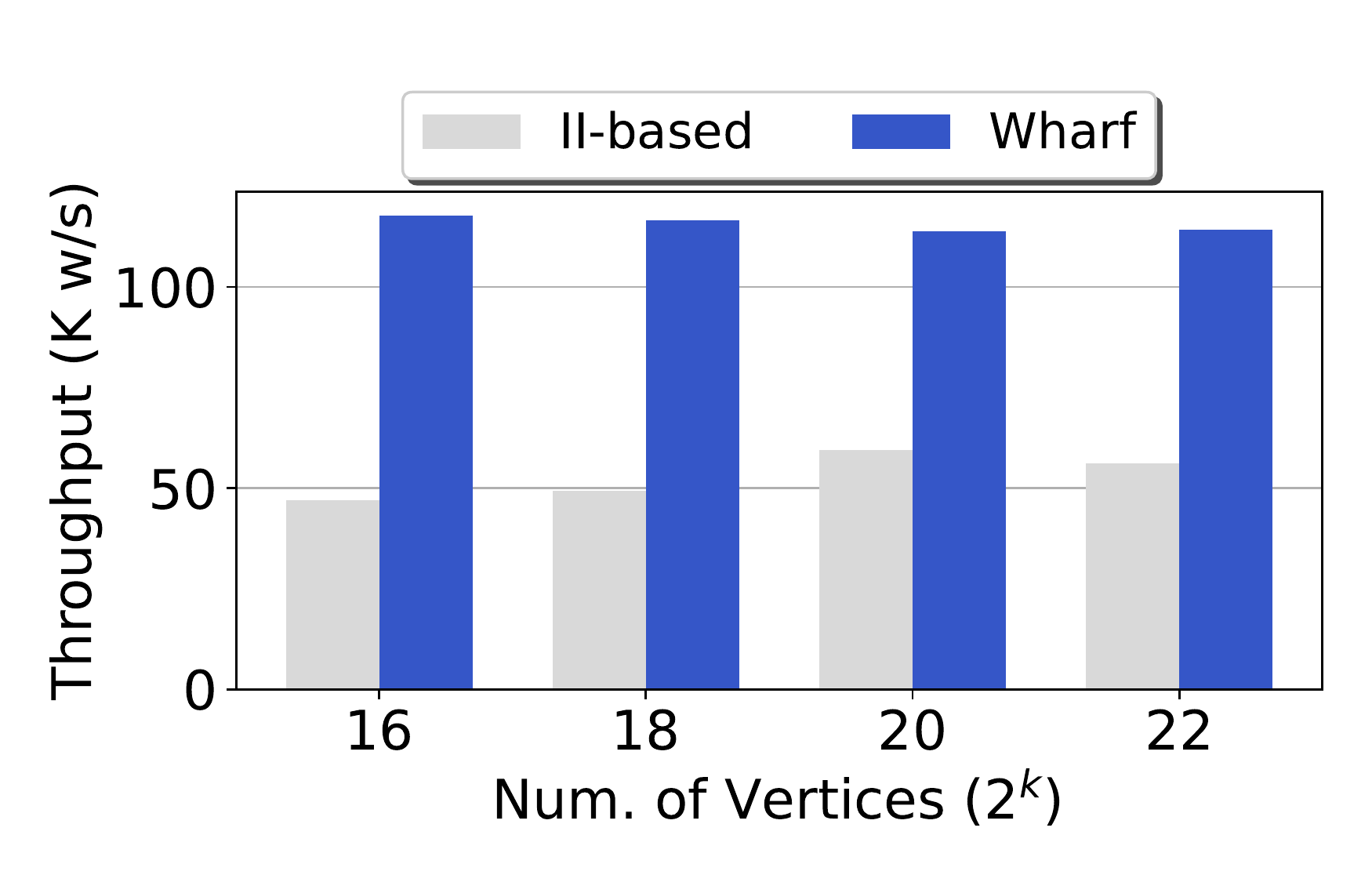}
    \caption{Throughput}
    \label{fig:er-throughput}
  \end{subfigure}
  \begin{subfigure}[t]{0.22\textwidth}
    \includegraphics[width=\textwidth]{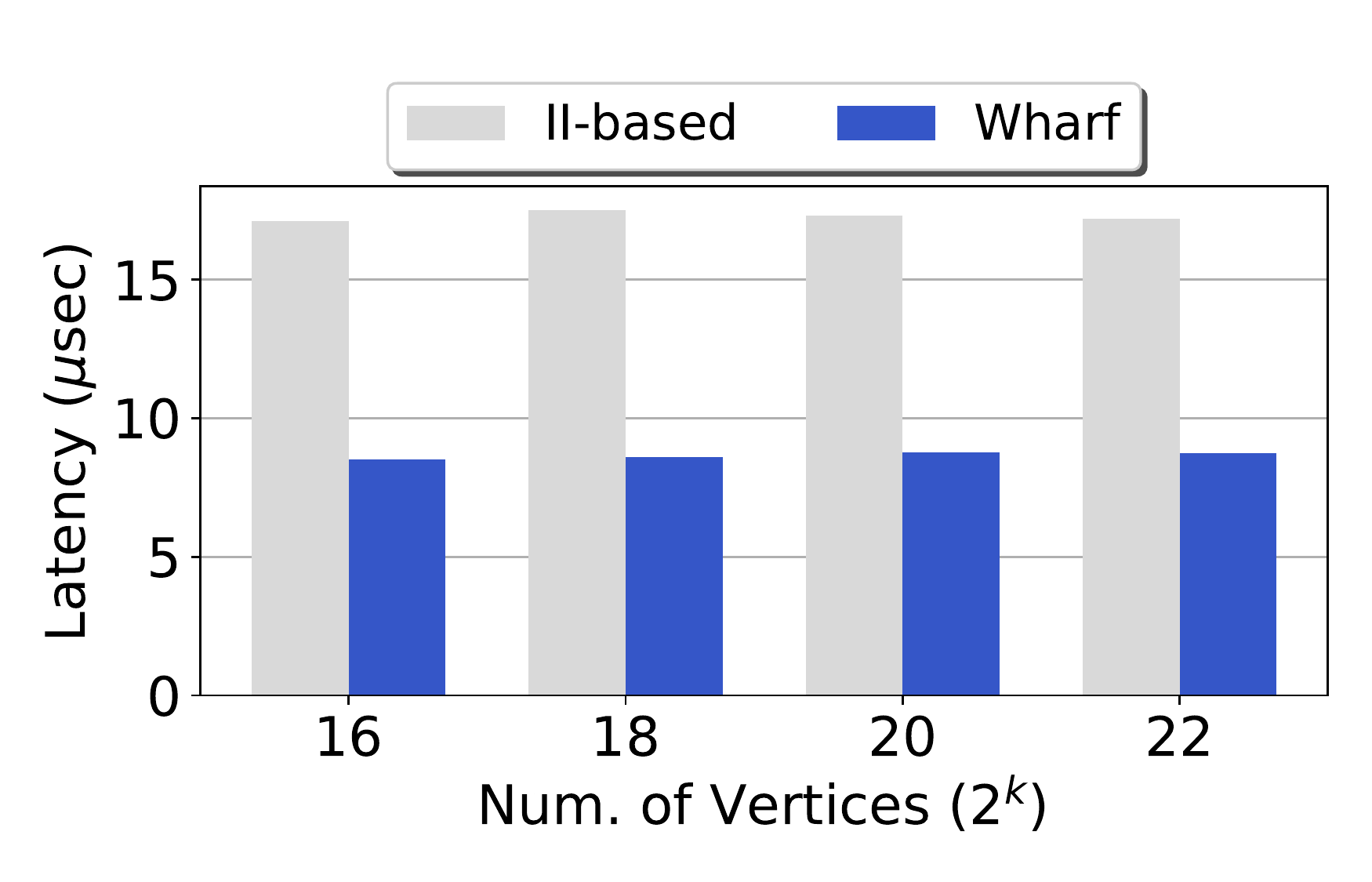}
    \caption{Latency}
    \label{fig:er-latency}
  \end{subfigure}
  \vspace{-0.2cm}
  \caption{Scalability as the graph size increases on $er$-graphs.}
\vspace{-0.4cm}
\end{figure}

\subsection{Performance under Data Skewness}
\label{subsec:skewness}
Next, we investigate the effect of graph skew on \wharf's performance in terms of
throughput and
memory footprint.
Specifically, we use a set of skewed graphs, $sg$-$s$, which all have $2^{20}$ vertices, where  we vary the skew factor $s=1, 3, 5, 7$. 
We set the batch size to $10K$ edges, and for each graph, we generate the edge updates using RMAT such that they follow the same distribution as the graph. 

Figure~\ref{fig:sg-throughput} depicts the throughput that \wharf and II-based achieve while performing walk updates. 
Recall that the black horizontal lines in the figure show the minimum throughput required to generate the walks from scratch.
We observe that \wharf has up to $\sim\!\!2\times$ better throughput than II-based. Actually, II-based not only has low throughput, but also needs more time to update the walks than generating them from scratch for $s \geq 3$. 
This is because the more skew in the graph the more often the high degree nodes appear in random walks, and thus, more random walks get affected.
Therefore, II-based falls short as it has to update the walk sequences and the walk index.   
Notice that the throughput of both \wharf and II-based decreases by $\sim\!\!18\%$ when going from $s=1$ to $s=7$, yet \wharf's throughput remains sufficiently high.

\begin{figure}[t]
  \begin{subfigure}[t]{0.23\textwidth}
    \includegraphics[width=\textwidth]{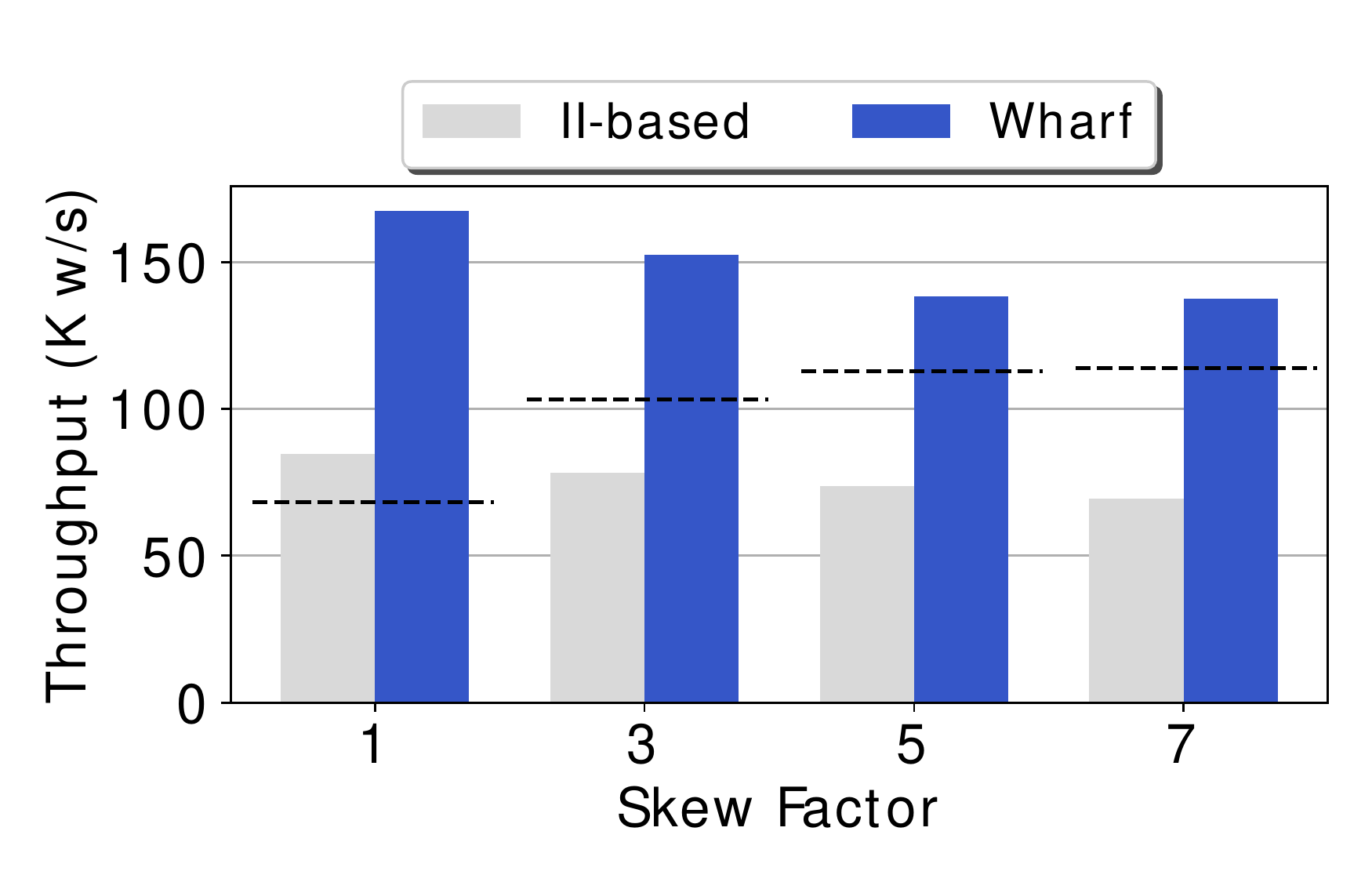}
    \caption{Throughput}
    \label{fig:sg-throughput}
  \end{subfigure}
  \begin{subfigure}[t]{0.23\textwidth}
    \includegraphics[width=\textwidth]{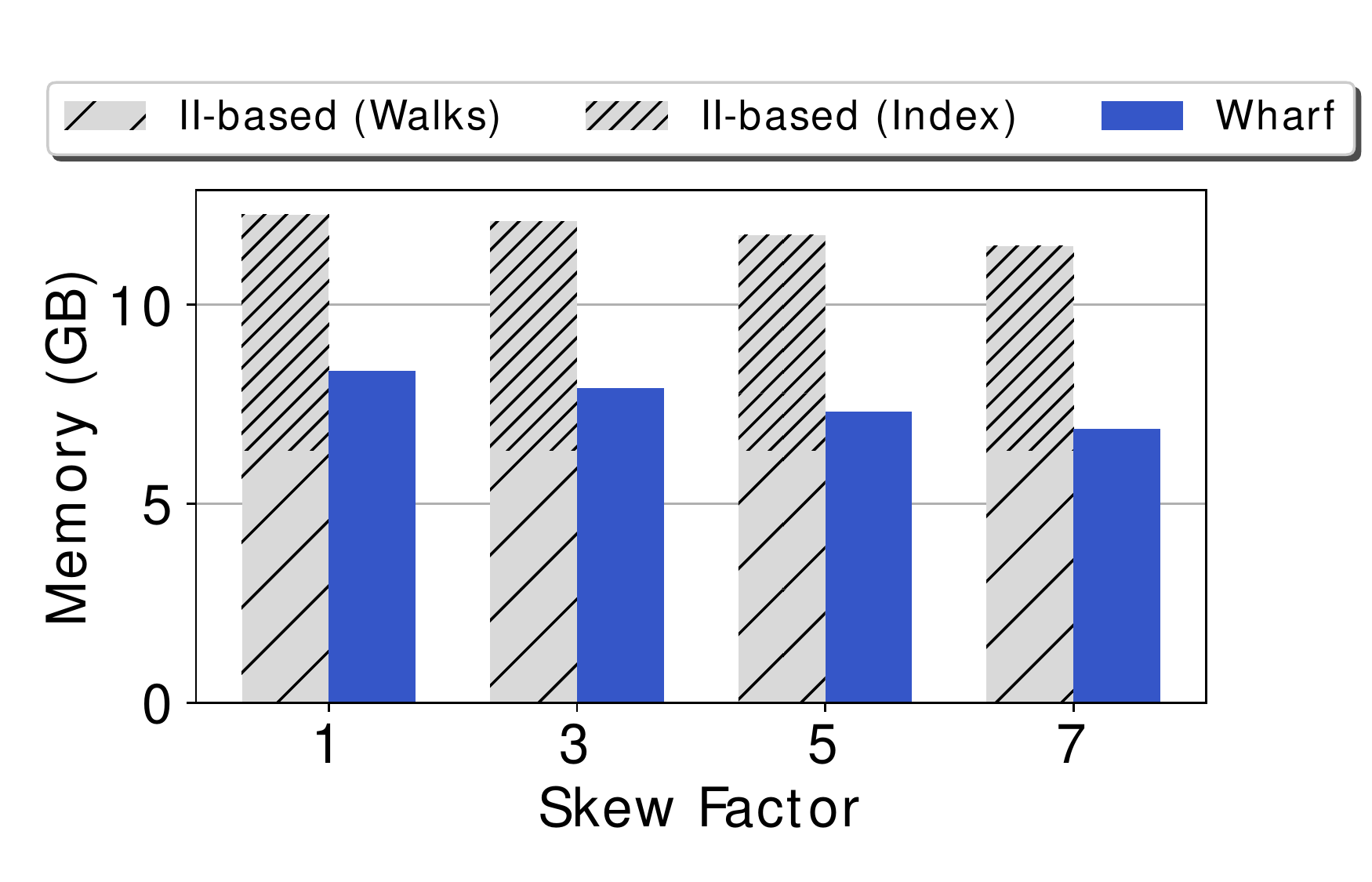}
    \caption{Memory footprint}
    \label{fig:sg-memory}
  \end{subfigure}
    \vspace{-0.2cm}
    \caption{Performance and space on the skewed $sg$-graphs.}
  \label{fig:sg-throughput-latency-memory}
  \vspace{-0.4cm}
\end{figure}

Figure~\ref{fig:sg-memory} shows the memory footprint.
We see that the higher the skew of the input graph the less space required to store the walks. 
This happens because a small number of vertices has an extremely high degree and appear many times in the majority of the random walks.
Therefore, the difference encoding in each chunk of walk-trees, in which \wharf stores the vertex ids of the walks, achieves better compression as the ids in a chunk mostly belong to high degree nodes and their neighbouring vertices. 
In contrast, II-based needs constant space for storing the walk sequences, while the inverted index space decreases when the skew increases but not as drastically as \wharf. 
Specifically, the memory footprint, between skew factors $s=1$ and $s=7$, drops by $\sim\!\!17.6\%$ in \wharf while only by $\sim\!\!6.4\%$ in II-based. 
In conclusion,\emph{ \wharf is robust to skew in terms of both throughput and space efficiency}.


\subsection{In-Depth Study}
\label{subsec:in-depth}



\noindent{\bf Range vs. Simple Search.} We proceed in exploring the benefits of the output-sensitive \textsc{FindNext} range search algorithm that \wharf uses when seeking a specific walk triplet inside a walk-tree.
As baseline, we disabled this range search and leave \wharf with the simple search that checks triplets by  scanning the entire walk-trees.


Figure~\ref{fig:throughput-node2vec-realgraphs-bsize10000-rvsssearch} illustrates the throughput improvement factor (IF) that \wharf's range search achieves when running node2vec, with parameters $p = 0.5$ and $q=2$, on all the real graphs.
We observe that the range search feature leads up to $3\times$ higher throughput than the baseline. 
Note that in smaller graphs, such as \textit{com-YouTube}, the gains from range search are higher than in larger graphs, such as \textit{soc-Livejournal}. This is because smaller graphs have less vertices in their walk-trees, which makes the range search faster as the constructed ranges are smaller`.
Furthermore, Figure~\ref{fig:throughput-node2vec-orkut-bsizevarying-rvsssearch} shows the throughput IF
for node2vec when varying the batch sizes of edge insertions on the large graph \textit{soc-LiveJournal}, where the walk-trees contain a huge amount of walk triplets.  
We get similar results for \textit{com-Orkut} but we omit them due to space limitations.
In this case, range search enables \wharf to achieve on average $\sim\!\!1.7\times$ higher throughput than the simple search. 
Note that the space overhead for the $\{min, max\}$ bounds in each walk-tree (necessary for our search range search) is negligible (less than $1\%$).
We thus conclude that \emph{the range search technique significantly contributes to the high throughput of \wharf with negligible space overhead}.

\begin{figure}[t]
  \begin{subfigure}[t]{0.23\textwidth}
    \includegraphics[width=\textwidth]{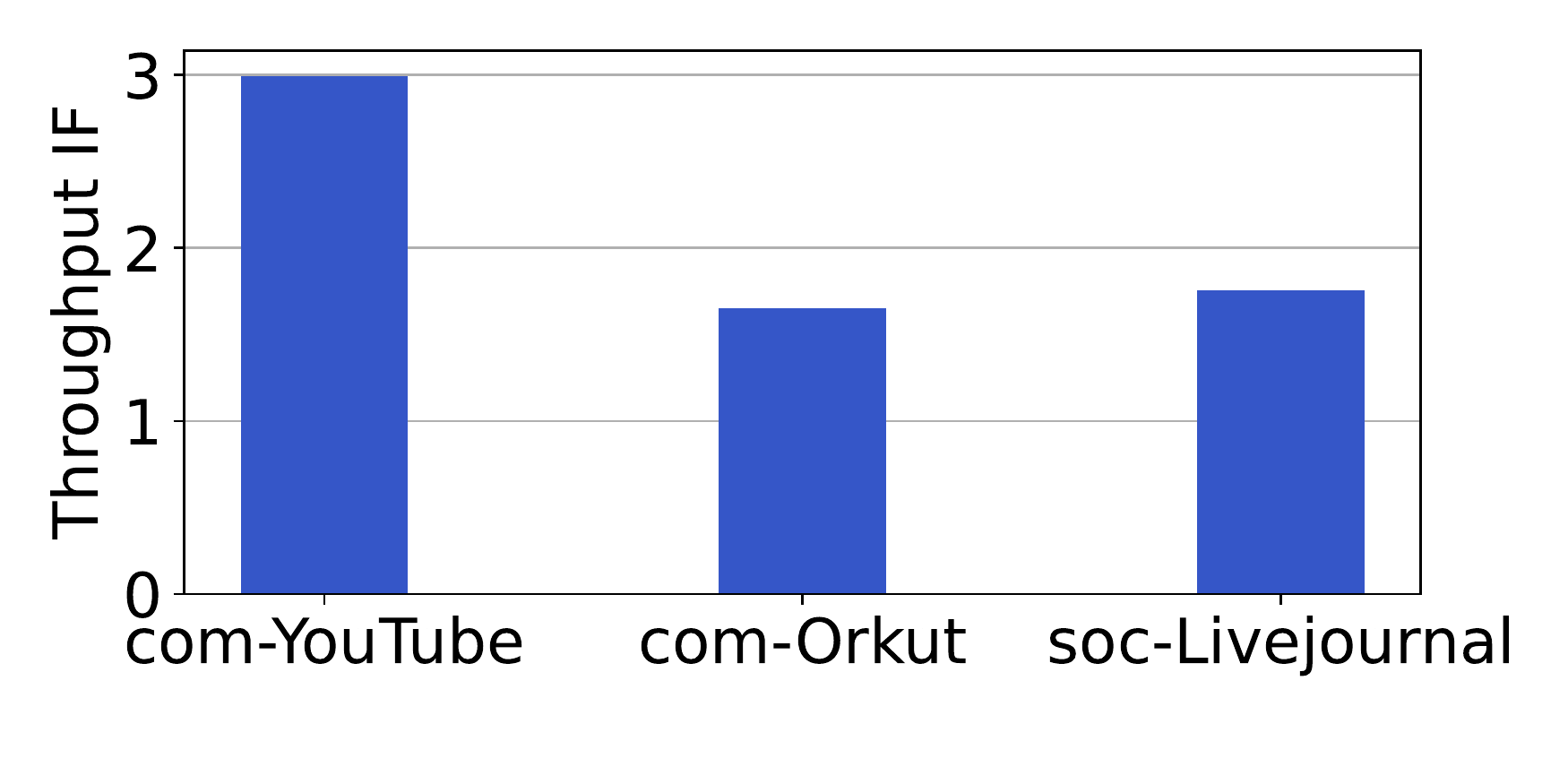}
    \caption{Real graphs, inserting $10K$ edges}
    \label{fig:throughput-node2vec-realgraphs-bsize10000-rvsssearch}
  \end{subfigure}
  \begin{subfigure}[t]{0.23\textwidth}
    \includegraphics[width=\textwidth]{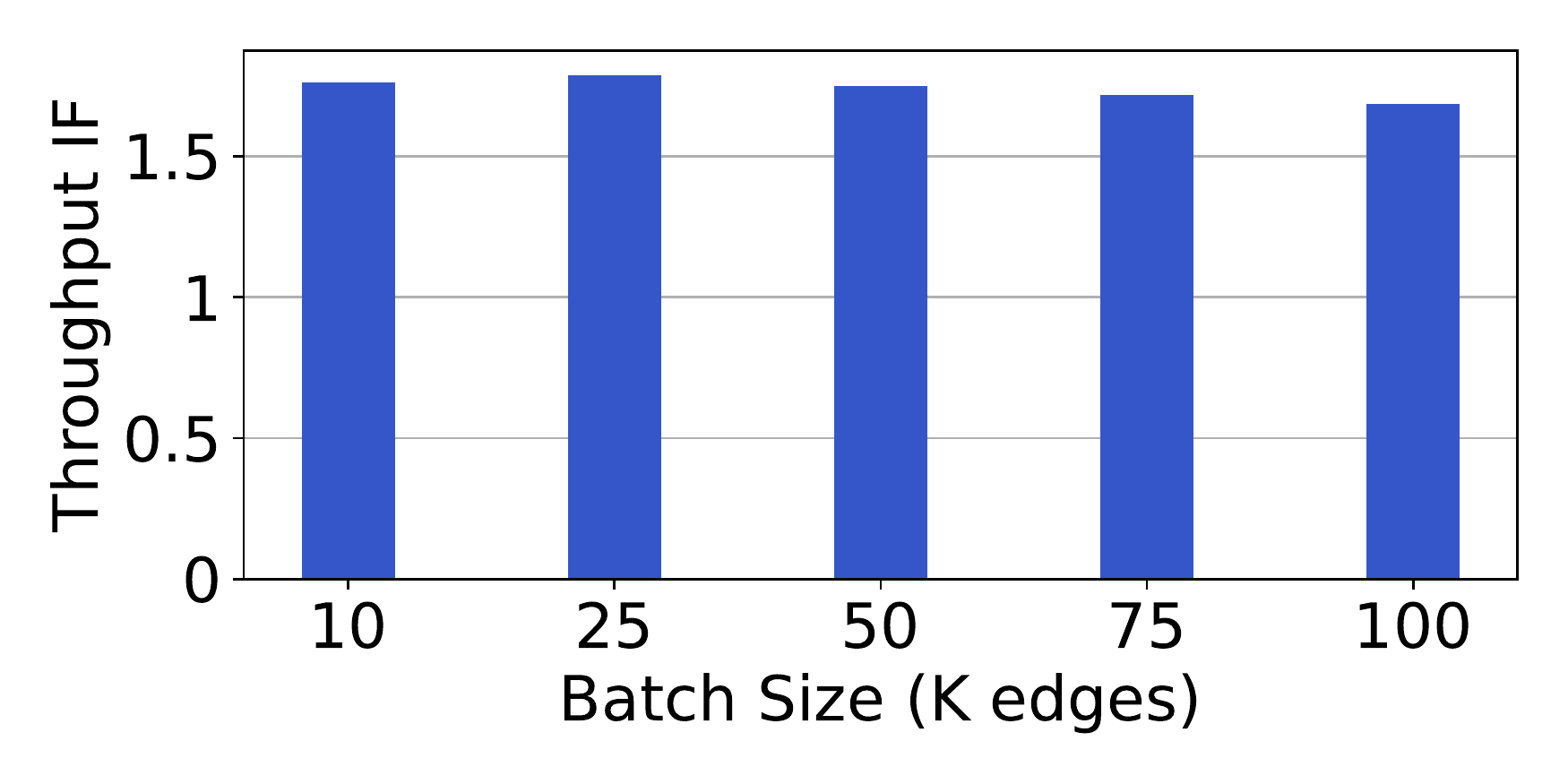}
    \caption{\textit{Livejournal}, varying batch size}
    \label{fig:throughput-node2vec-orkut-bsizevarying-rvsssearch}
  \end{subfigure}
  \vspace{-0.2cm}
  \caption{Throughput improvement factor (IF) of \wharf when using range over simple search for node2vec.}
  \label{fig:livejournal-sVSrsearchy}
  \vspace{-0.4cm}
\end{figure}

\vspace{0.1cm}
\noindent{\bf Benefits of Difference Enconding.} We now explore the impact of difference encoding (DE) on \wharf's throughput and memory footprint.
For this, we disabled the DE in \wharf and used the resulting variant as our baseline.
We inserted 10 batches of $10$K edges and measured the average throughput and the memory footprint after the merge operation for all our real graphs. We observed that \wharf needs up to $1.4\times$ less space to store the walks than when not using DE. 
It also achieves quite a similar throughput for all real graphs that is within $5\%$ as the one achieved when not having DE. 
For instance, on the LiveJournal dataset the throughput with DE is $\approx207.4$K walks/second, whereas without it is $\approx214.6$K walks/second. 
We thus conclude that \textit{compression via difference encoding helps improve the memory footprint of \wharf, however, its performance does not stem from compression but from our proposed techniques.}

\vspace{0.1cm}
\noindent{\bf Vertex Id Distribution.} We now explore the effect that the vertex id distribution has to the space needed to store random walks. 
Specifically, we used our $er$-$18$ graph that has $262,144$ vertices as the initial graph $G_1$. 
In $G_1$ the vertex ids are fully clustered, i.e.,~they range from $0$ to $262,143$, as produced by TrillionG~\cite{sigmod17-trillionG}. 
From $G_1$, we produced $G_{2-x20}$ by multiplying the vertex ids by $20$ to make the ids of the graph non-clustered, yet ordered.  
Additionally, we created two more graphs out of $G_1$, namely, $G_{3-r1M}$ and $G_{4-r5M}$ where we reassigned a unique random id to each vertex drawn from the $[0-1M)$ and $[0-5M)$ ranges, respectively. 
We generated a  walk corpus of $10$ walks per vertex of length $80$ each for all four aforementioned graphs. 
We observed that the space that \wharf needs to store the walks for $G_1$ is $1.553$ GB, for $G_{2-x20}$ is $1.547$ GB, for $G_{3-r1M}$ is $1.553$ GB, and for $G_{4-r5M}$ is $1.547$ GB.
This shows that the delta encoding scheme ensures that \emph{the space \wharf needs to store the walks is not affected by the vertex id distribution}.

\subsection{Effectiveness of Downstream Tasks}
\label{subsec:accuracy}


Lastly, we measure the accuracy of a downstream vertex classification task and a Personalized PageRank (PPR) task to show the effectiveness of \wharf to maintain random walks. 
For the former, we implemented an {\em incremental learning} approach that uses \wharf: it builds a predictive model, after each graph update (snapshot), from embeddings that use the walks produced by \wharf. 
For the latter, we implemented~\cite{vldb10-goel} in \wharf for producing and updating the walks used for approximating PPR scores.
We used the \textit{\textbf{Cora}} dataset for both tasks, which is an undirected citation graph of $2,708$ vertices~\cite{arxiv19-incremental-node2vec} and $5,429$ edges where each vertex (i.e., paper) belongs to one of 7 categories, such as ``neural networks''.
We ingested at every timestep a new batch of edges and incorporated them into the graph.
We fixed the batch size to $250$.

\header{Vertex classification}
%
%
As baselines for the vertex classification task, we considered (i)~the {\em ideal learning} case, i.e.,~learning a new model at every single snapshot, and (ii)~the {\em periodic learning} case, i.e.,~learning a new model every $k$ snapshots (we use $k =5,10$). 
Both ideal and periodic learn embeddings using random walks computed from the scratch.
For the incremental learning case (which is based on \wharf), 
\wharf updates the walks after each batch insertion so that they remain statistically indistinguishable. We, then, incrementally refine 
the embeddings 
using \textit{yskip}~\cite{emnlp17-yskip} 
with default DeepWalk parameters (i.e.,~$n_w = 10$, $l = 80$), and trained $128$-sized embedding vectors.
We used LogisticRegression for the 
classification and report the average $F1$ score of three runs. 
Figure~\ref{fig:fig:accuracy-cora-deepwalk} 
shows
that the incremental learning achieves overall the same accuracy as the ideal learning, demonstrating the high effectiveness of \wharf to maintain high-quality random walks.
Note that in the periodic learning scenario the accuracy drops significantly in between the snapshots where re-training takes place. 
	The larger the period (e.g.,~$k=10$) the lower the accuracy stays.
	These results are aligned with Figure~\ref{fig:motivation-ge}, where we witnessed a large decrease in the accuracy of a link prediction task if we do not keep the embeddings up-to-date. 
	Such drops in accuracy can have large negative impact in the underlying ML tasks, especially for high-stakes applications, such as fraud detection.

We can then also conclude that having statistically indistinguishable random walks is crucial.

\header{Personalized pagerank}
For the PPR task, we considered a \textit{static} variant of~\cite{vldb10-goel} as baseline, which reuses the existing random walks instead of updating only the affected walks at every snapshot.
We generated $10$ walks per vertex with a restart probability of $0.2$, and report the Symmetric Mean Average Error (SMAPE) between~\cite{vldb10-goel} and the static variant (Figure~\ref{fig:fig:accuracy-cora-ppr}).
	We observe that as more graph snapshots arrive the SMAPE constantly increases. 
	In fact, even after the first snapshot arrives, the error is already greater than 40\%. 
	These results are aligned with Figure~\ref{fig:motivation-ppr}, yet, because of the smaller batch sizes we used here, 
	the error gradually increases. 
These results confirm the vertex classification results: keeping random walks statistically indistinguishable is crucial for the downstream tasks.

\begin{figure}[t!]
   \vspace{-0.4cm}
   \begin{subfigure}[t]{0.23\textwidth}
    \includegraphics[width=\linewidth]{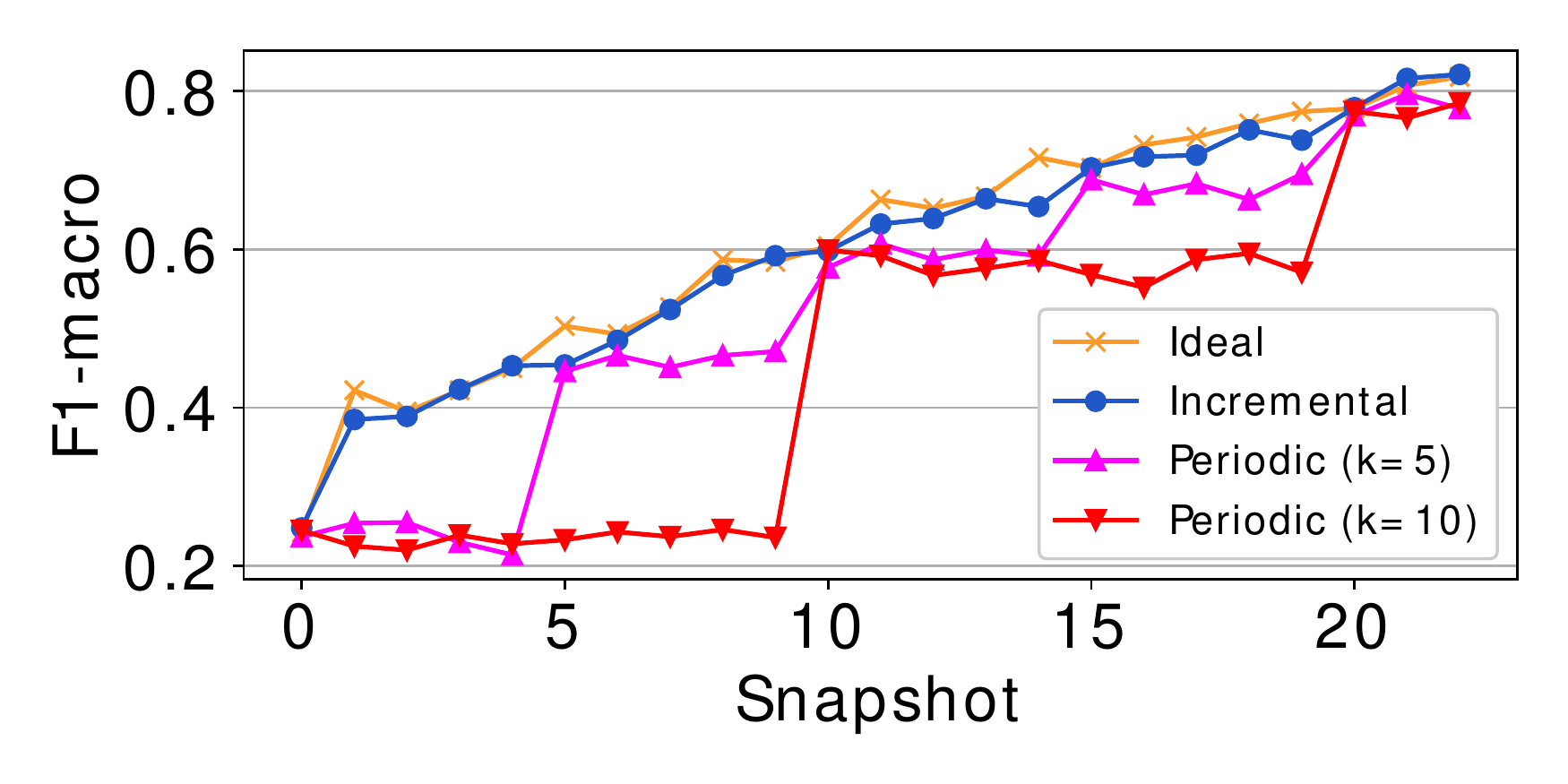}
    \caption{Vertex Classification (DeepWalk)}
    \label{fig:fig:accuracy-cora-deepwalk}
  \end{subfigure}
  \begin{subfigure}[t]{0.23\textwidth}
    \includegraphics[width=\linewidth]{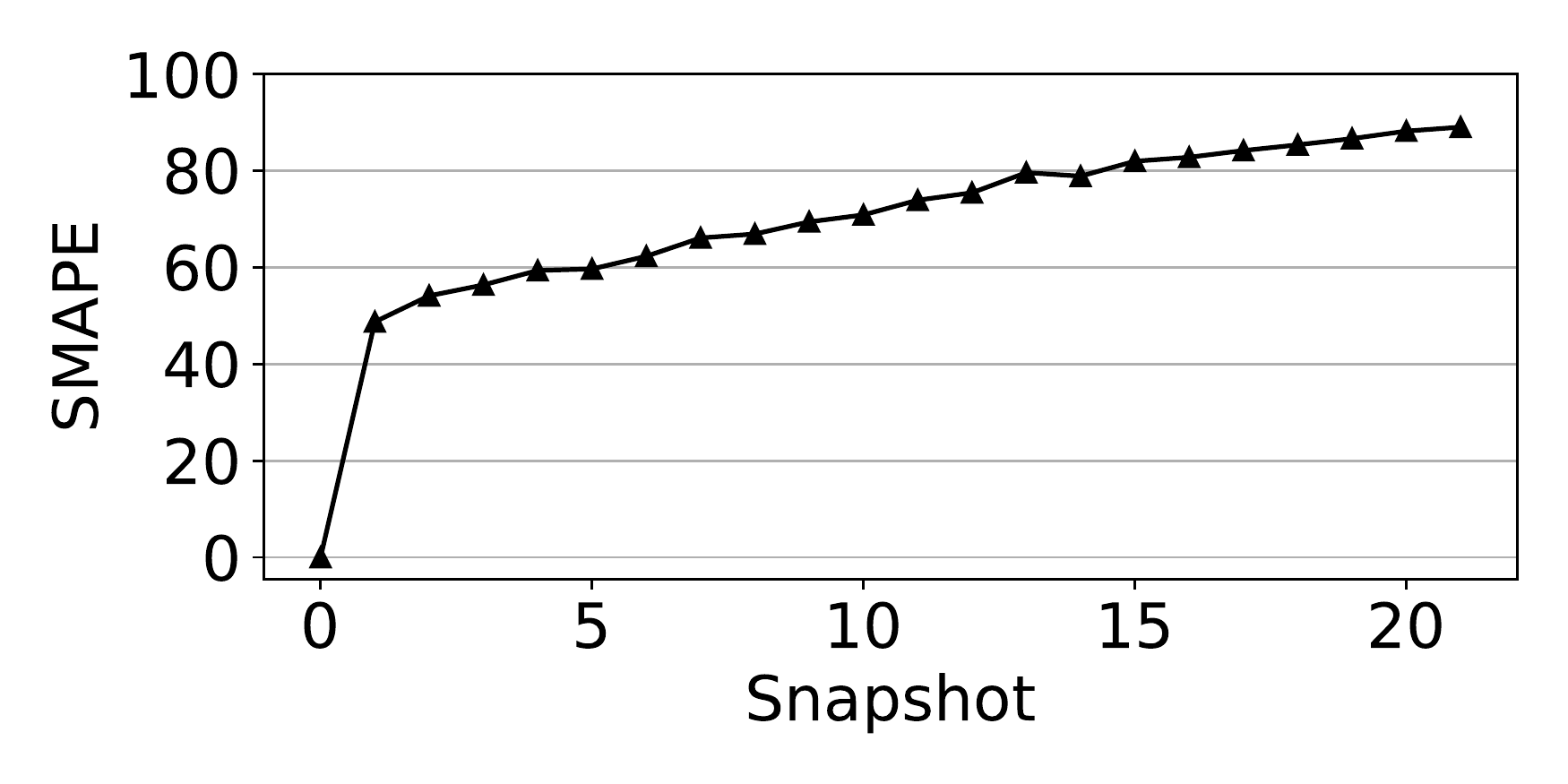}
    \caption{Personalized PageRank}
    \label{fig:fig:accuracy-cora-ppr}
  \end{subfigure}
 \vspace{-0.2cm}
 \caption{Accuracy of 
 downstream tasks 
 on \textit{Cora}.}  
  \label{fig:livejournal-sVSrsearchy}
  \vspace{-0.4cm}
\end{figure}

\section{Related Work}
\label{sec:related-work}

\noindent \textbf{Streaming Graph Systems.} Aspen~\cite{pldi19-aspen} is a state-of-the-art streaming graph processing framework that is built for single-node execution and uses the novel $C$-tree data structure for representing the graph in a compressed manner, while provably achieving very low latency when processing streaming graph workloads. 
Teceo~\cite{vldb21-teceo} is a recently introduced system that can store and analyze dynamic graphs in main-memory with transactional capabilities. 
Teceo presents a data structure that maintains a dynamic graph. This structure is based on sparse arrays and a fat tree where the graph is kept at the end. 
The system does not offer any compression primitives, which is highly prohibitive for ML applications we tackle, such as graph embeddings.  
DZig~\cite{eurosys21-dzig} is a quite recently proposed high-performance system that enables performant graph processing in the presence of sparsity and at the same time guarantees bulk-synchronous-parallel semantics. 
Tripoline~\cite{eurosys21-tripoline} is a system built on top of Aspen that evaluates queries without their a priori knowledge (e.g.,~BFS) in an incremental way. 
It generalizes the incremental graph processing by using graph triangle inequalities and natively supports high-throughput graph updates of low cost. 
Other well-known graph streaming systems exists as well~\cite{icde15-llama, hpec12-stinger} but none of them deals with random walks.

\noindent\textbf{Random Walk Systems.} DrukarMob~\cite{recsys13-drunkardmob} is a single machine out-of-core system for calculating first-order (uniform) random walks, but it is not main memory-based. 
Shao et al.~\cite{sigmod20-lei} propose a memory-aware 
random walk framework that dynamically assigns different sampling methods to minimize the cost of random walks within the memory budget. 
KnightKing~\cite{sosp19-knightking} is a distributed system for computing random walks on static graphs based on a walker-centric computation model, which is able to express various walk algorithms. 
ThunderRW~\cite{vldb21-thunderrw} is a single-node system that conducts in-memory random walks by devising a step-centric computation model, which hides memory access latency by executing multiple queries in an alternating manner. 
In contrast to the above systems, \wharf is designed for streaming graphs, does not impose any main memory constraints, and 
supports 
walks of 
any order.

\noindent \textbf{Dynamic GRL.} 
Barros et al.~\cite{survey21-barros} categorize random walk-based GRL methods on dynamic graphs into: (i) random walks on snapshots, (ii) evolving random walks, and (iii) temporal random walks. 
In the first category, random walks are re-computed at every snapshot so that embeddings are learned from scratch.
\wharf falls into the second category, where random walks are not recomputed from scratch after every graph update, but they are updated along with the embeddings of the affected vertices~\cite{cnta19-evonrl, arxiv19-incremental-node2vec, bigdata18-madhavi}.
Yet, \wharf stores, indexes, and updates the walks more efficiently than the competitor approaches as we show in the experimental section. 
The third category, contains methods that consider temporal walks, i.e.,~the temporal flux is respected during their creation~\cite{www18-ctdne, cikm20-tdgraphembed, asonam18-winter}. However, \wharf does not consider the temporal aspect of edges. 


\noindent\textbf{Dynamic PageRank.}
Bahmani et al.~\cite{vldb10-goel} present a method for calculating Personalized PageRank (PPR) scores via precomputing and storing 
random walks for each node in the graph. 
The stored walks are not indexed leading to a full scan of walks for each incoming edge update. 
In contrast, \wharf's structure offers an 
 index on the walks which leads to faster walk updates.
Mo et. al.~\cite{cikm21-agent} present the Agenda framework for fast and robust Single Source PPR queries on evolving graphs. 
Yet, Agenda does not store the entire random walks in main memory, which might require recomputing walks from the scratch.
Jiang et al.~\cite{vldb17-reads} propose a walk indexing scheme for calculating SimRank in dynamic graphs. However, the authors designed an index specifically for simrank-aware walks and hence one cannot use it out-of-the-box in our setting. 

\section{Conclusion}
\label{sec:conclusion}

We tackled the problem of computing and maintaining random walks up-to-date in streaming graphs ({\em streaming random walks}).
We presented \wharf, a system that produces and updates random walks in a streaming fashion as well as stores them succinctly.
\wharf represents walks concisely by coupling compressed purely functional binary trees and pairing functions and updates the walks efficiently by pruning the search space leveraging the ordering properties of pairing functions.
Our experiments show that \wharf can incrementally update 
walks with up to $2.6\times$ higher throughput and up to $2\times$ lower latency than inverted index-based baselines.  

\bibliographystyle{ACM-Reference-Format}
\bibliography{sample}



\begin{appendix}

\section{Merge Policies}
\label{appendix:merge-policies}
Let us now briefly discuss the different policies for merging walk-trees to evict obsolete walk-triplets from the walk corpus. 
In all the experiments we showed in Section~\ref{sec:experiments}, \wharf used the on-demand policy, where the merge 
took place at the last batch. 
We have also experimented with a different policy of merging after each batch (eager policy). 
In general, we observed that \wharf achieved on average $\sim\!\!1.8\times$ higher throughput when using the (default) on-demand policy than when using the eager one, and $\sim\!\!2.6\times$ higher than II-based. 
Certainly, this comes at the cost of the memory footprint as keeping the different walk-tree versions increases the space demands. 
For instance, the on-demand policy leads to a cumulative memory footprint $1.6\times$ larger after inserting $4$ batches compared to the constant memory footprint of the eager policy, as shown in Figure~\ref{fig:merge-policy}. 
Note that in the illustrated experiment we
insert $5$ batches of $25K$ edges on \textit{com-Orkut}. 
Thus, \textit{there exists a throughput-memory trade-off: one can achieve significantly higher throughput at the price of a higher cumulative memory footprint by merging less frequently.}   
We leave a deeper study of this trade-off  for future work. 
Specifically, it would be particularly interesting to see how one may interleave the on-demand with the eager policy depending on the throughput Service Licence Agreements (SLAs) and the memory capabilities of the system \wharf is running on. 

\begin{figure}[t]
 \begin{subfigure}[t]{0.23\textwidth} 
   \includegraphics[width=\textwidth]{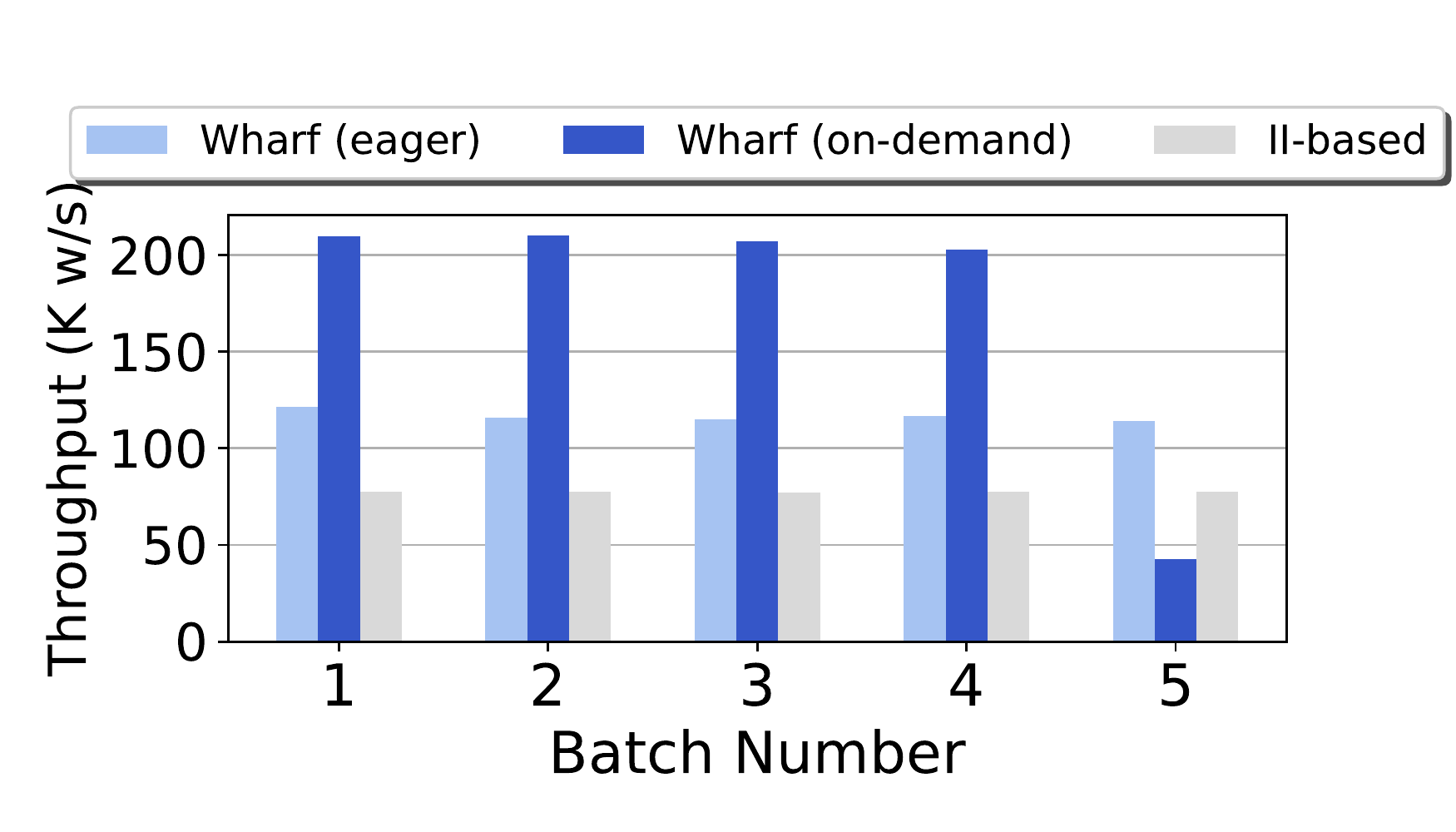}
   \caption{Throughput}
   \label{fig:throughput-merge-policy}
 \end{subfigure}
 \begin{subfigure}[t]{0.23\textwidth} 
   \includegraphics[width=\textwidth]{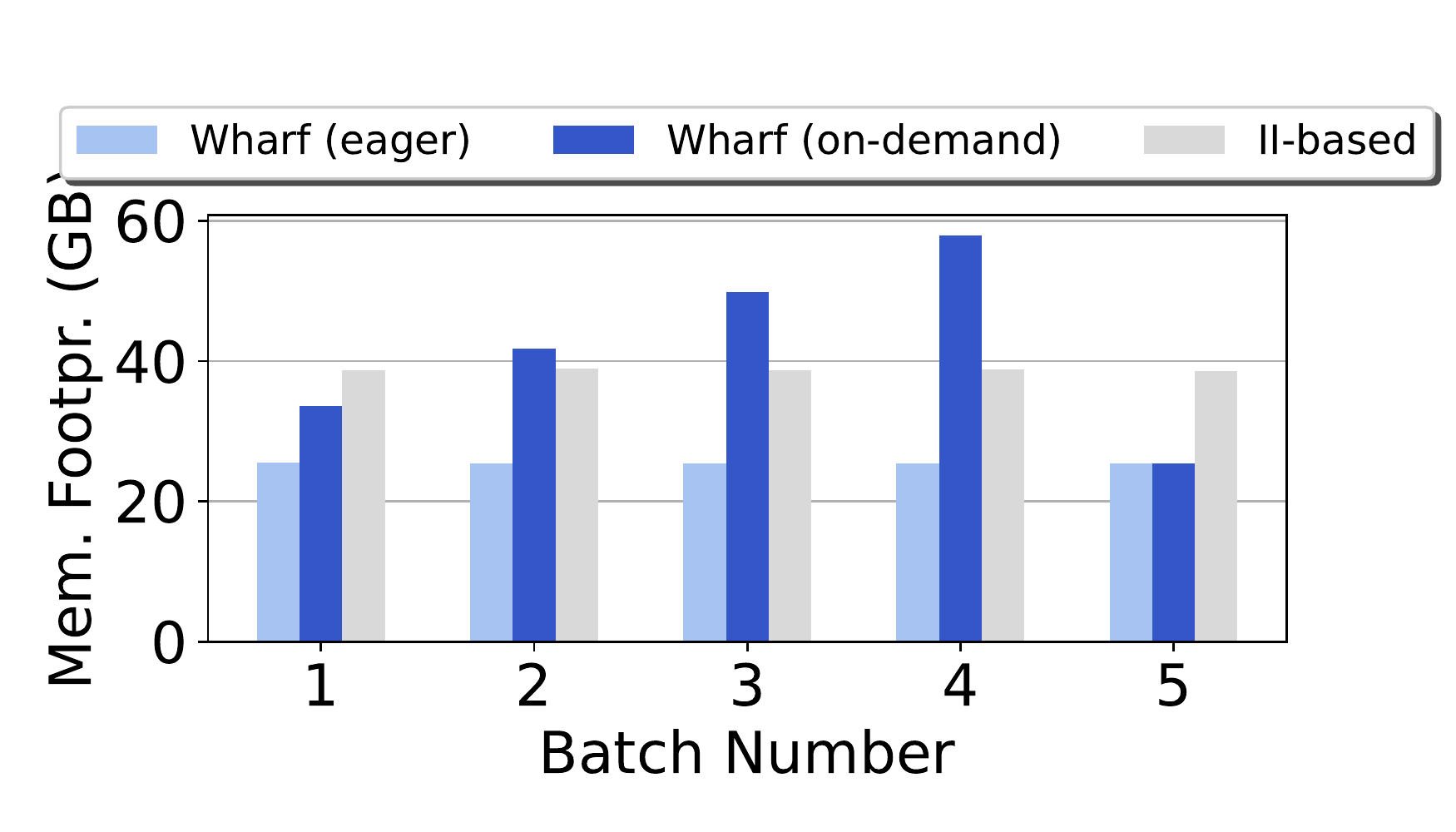}    
   \caption{Memory footprint}
   \label{fig:memory-merge-policy}
 \end{subfigure}
 
 \caption{Comparison between 
 the on-demand and eager merge policies while inserting $5$ batches of $25K$ edges on \textit{com-Orkut}.}
 \label{fig:merge-policy}
\end{figure}

\section{Distribution of Minimum Affected Position in Each Affected Walk}
\label{appendix:pmin-distribution}

Let us now shed some more light on how the batch sizes influences \wharf's throughput.  
Figure~\ref{fig:distribution-pmin} illustrates the distribution of minimum affected positions, i.e.,~$p_{min}$, in the affected walks for different batch sizes in \textit{com-Orkut}. 
We observe that as the batch size increases not only does the number of affected walks increases, but also more and more 
walks are affected at an earlier point of their walk sequence. 
This leads to more work for updating the random walks. 
Thus, we conclude that \textit{not only the number of affected walks that influences the achieved throughput, but also the minimum affected positions of the affected walks.}

\section{Throughput varying Walk Length}
\label{appendix:throughput-varying-l}

Here, we present another experiment that illustrates the overall performance of \wharf and our baselines in terms of throughput when we vary the length of the random walks we generate and maintain. 
Specifically, we inserted $10$ batches of $10K$ edges in the \textit{soc-Livejournal} graph aiming to measure the average throughput of walk updates for walk length values $l \in \{5, 10, 15, 40, 80, 120\}$. 
In Figure~\ref{fig:throughput-varying-l} we show the achieved throughput for the aforementioned walk length values. 
As we see, the throughput of \wharf as well as of its competitors is decreasing as the walk length increases. 
The performance gap of II-based compared to \wharf is bigger for greater
walk lengths, i.e., for $l \in \{40, 80, 120\}$, whereas it is smaller for smaller walk lengths, i.e., $l \in \{5, 10, 15\}$. 
Yet, we can trivially see that \textit{\wharf always retains a greater throughput than both II-based and Tree-based.}

\begin{figure}[t]
 \begin{subfigure}[t]{0.23\textwidth} 
   \includegraphics[width=\textwidth]{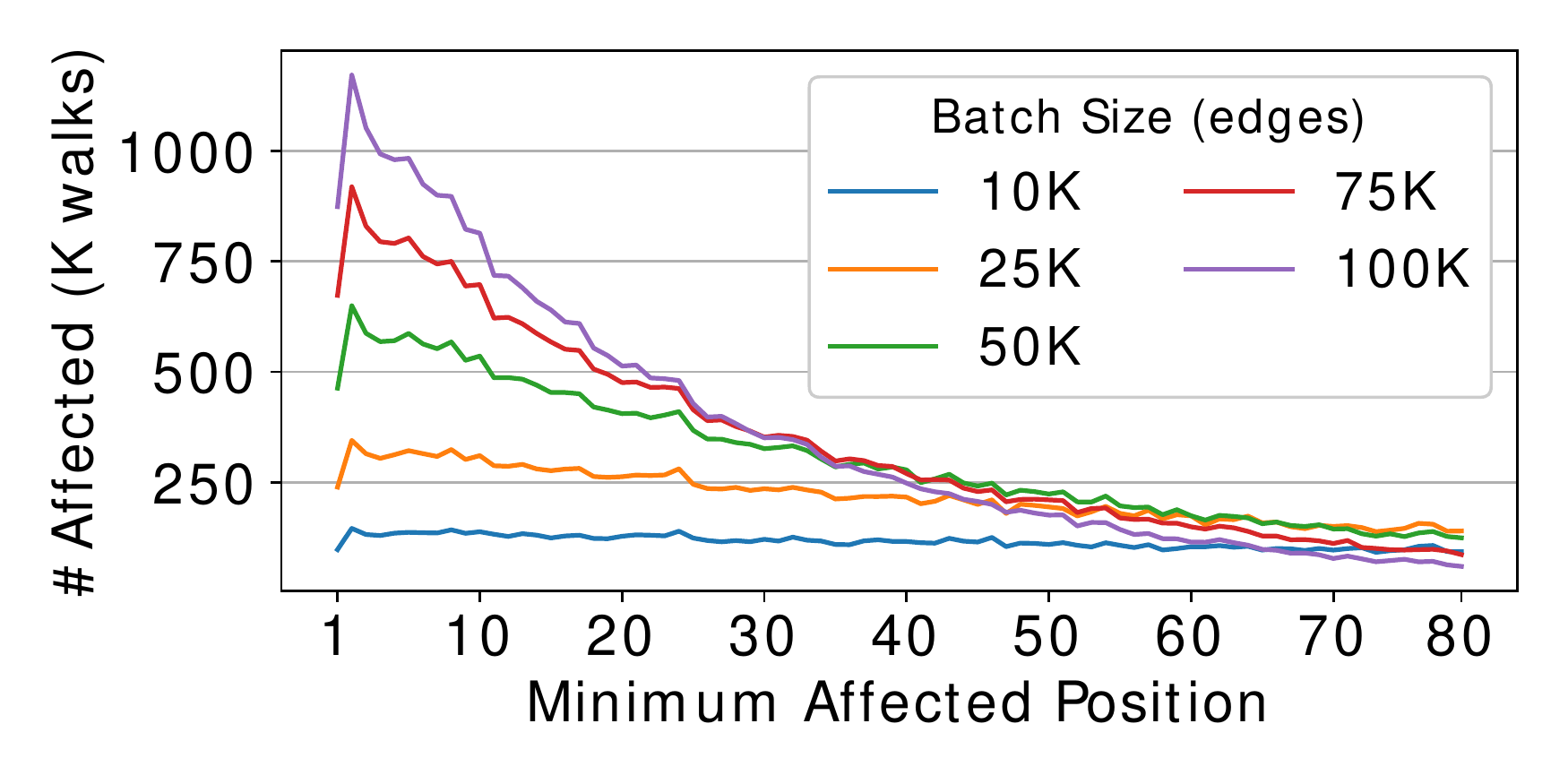}
   \caption{Distribution of minimum affected position in each affected walk}
   \label{fig:distribution-pmin}
 \end{subfigure}
 \begin{subfigure}[t]{0.23\textwidth} 
   \includegraphics[width=\textwidth]{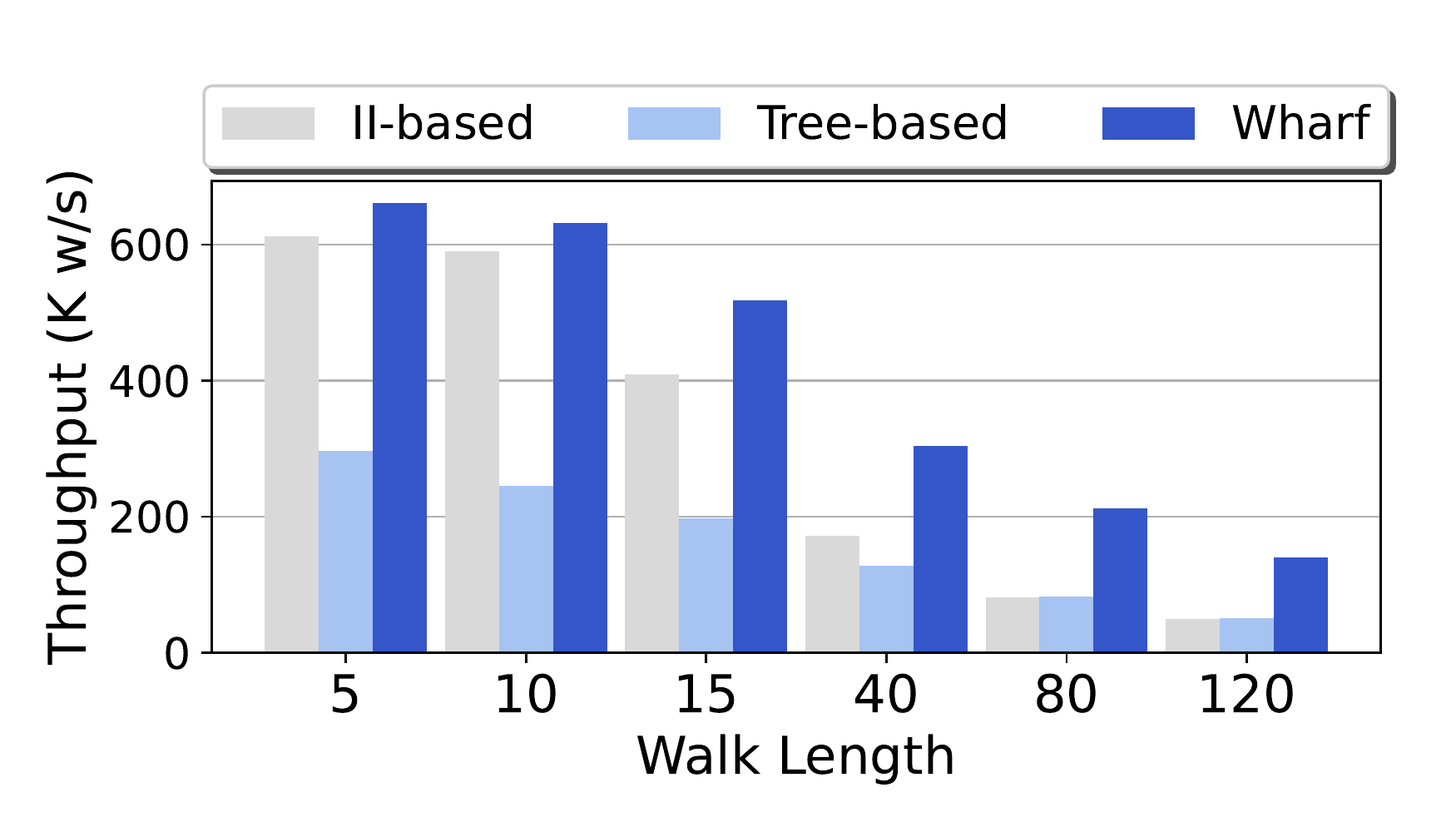}    
   \caption{Throughput varying $l$}
   \label{fig:throughput-varying-l}
 \end{subfigure}
 
 \caption{Effect of the batch size and of the walk length.}
 \label{fig:throughput-varying-l-batch-size-effect}
\end{figure}

\end{appendix}

\end{document}